\documentclass[a4paper,fleqn]{cas-dc}

\usepackage[numbers]{natbib}
\usepackage{cuted}   
\usepackage{amsmath, amsfonts, amssymb, amsthm, amstext, mathtools, bm, bbm, empheq, comment}
\usepackage{needspace}  
\usepackage{placeins} 
\usepackage[linesnumbered,ruled,vlined]{algorithm2e}
\SetKwInput{KwInput}{Input}
\SetKwInput{KwOutput}{Output}

\usepackage{graphicx, subcaption, tikz, float, textcomp, booktabs, tabularx, multirow, array}
\usepackage{enumitem}
\usepackage{xcolor}
\usepackage{adjustbox}
\usepackage{ragged2e}
\usepackage{eso-pic}
\usepackage[capitalise]{cleveref}
\usepackage{bibentry}
\usepackage{soul}

\definecolor{darkblue}{RGB}{0,0,139}
\newcolumntype{B}{>{\bfseries\textcolor{darkblue}}c}
\newcolumntype{C}{>{\bfseries\textcolor{darkblue}}c}
\setlist{nosep}

\newtheorem{theorem}{Theorem}[section]

\newtheorem{lemma}[theorem]{Lemma}

\newtheorem{remark}[theorem]{Remark}
\newtheorem{proposition}[theorem]{Proposition}
\newtheoremstyle{example}
  {.3\baselineskip}{.3\baselineskip}{\normalsize}{0pt}{\bfseries}{.}{5pt plus 1pt minus 1pt}{}
\theoremstyle{example}

\newcommand*\E{\mathop{}\!\textnormal{e}}

\newif\ifreviewcomments
\reviewcommentstrue

\newcommand{\Rp}{\mathbb{R}_{+}}

\newcommand{\Price}{S}

\newcommand{\MarkSpace}{K}                       
   
\newcommand{\persec}{\ensuremath{\mathrm{s}^{-1}}}
\newcommand{\perhour}{\ensuremath{\mathrm{h}^{-1}}}

\NewDocumentCommand{\revpt}{ O{} m m }{%
  [\ifreviewcomments
    \textcolor{red}{\textbf{Point~#2:}}~
    \textcolor{blue}{#3}%
  \fi]
}

\shorttitle{Optimal Execution in Intraday Energy Markets}
\shortauthors{Chatziandreou \& Karbach}

\title[mode=title]{Optimal Execution in Intraday Energy Markets under Hawkes Processes with Transient Impact}

\author[]{Konstantinos Chatziandreou}[]
\ead{k.chatziandreou@uva.nl}

\author[]{Sven Karbach}[]
\ead{sven@karbach.org}

\affiliation[]{organization={Korteweg-de Vries Institute for Mathematics and Informatics Institute, University of Amsterdam},
                addressline={Science Park 105-107},
                city={Amsterdam},
                postcode={1098 XG},
                country={Netherlands}}

\begin{document}
\begin{abstract}
This paper investigates optimal execution strategies in intraday energy markets through a mutually exciting Hawkes process model. Calibrated to data from the German intraday electricity market, the model effectively captures key empirical features, including intra-session volatility, distinct intraday market activity patterns, and the Samuelson effect as gate closure approaches. By integrating a transient price impact model with a bivariate Hawkes process to model the market order flow, we derive an optimal trading trajectory for energy companies managing large volumes, accounting for the specific trading patterns of these markets. A back-testing analysis compares the proposed strategy against standard benchmarks such as Time-Weighted Average Price (TWAP) and Volume-Weighted Average Price (VWAP), demonstrating substantial cost reductions across various hourly trading products in intraday energy markets.
\end{abstract}


\maketitle

\section{Introduction}
Intraday markets play a crucial role in developed energy markets, with trading volumes reaching new peaks in recent years, as observed in the German market in 2024~\cite{EPEXSpotSE}. One of the primary drivers behind this growth is the increasing share of renewable energy, which requires market participants to balance their positions in short-term intraday markets due to the variable and uncertain production from sources like wind and solar. This behavior results in higher trading volumes and elevated price volatility as delivery time approaches, a well-documented phenomenon in commodity and energy markets known as the Samuelson effect~\cite{c706bac4-9c99-3ef5-b2ae-197e56e4542c}.

Empirical studies, including~\cite{Luckner} and \cite{f318e8a6-3566-30d7-a5ea-473f21c64fcd}, have validated the Samuelson effect in short-term energy markets. These studies highlight an exponential rise in order book activity as gate closure nears, with a surge in market and limit order arrivals driven by participants' urgency to balance positions. Such dynamics pose significant operational challenges for energy companies, particularly when managing substantial trading volumes arising from unexpected outages, fluctuations in renewable generation, or strategic trading decisions. The central question for market participants is how to execute these trades efficiently while minimizing costs and mitigating market impact. 

This paper addresses this challenge by developing optimal execution strategies for renewable energy producers in intraday energy markets. The proposed framework aims to minimize trading costs by leveraging a price impact model embedded within a limit order book setting using a mutually exciting Hawkes process to model market order flows. This study provides practical strategies for energy companies active in intraday energy markets by explicitly accounting for the unique trading activity patterns in these markets. In particular, it captures the exponential increase in trading volume approaching gate closure, followed by a sharp decline after market decoupling, where one hour before delivery, international orders are no longer coupled with the German market, and 30 minutes before delivery, trading transitions to local, within-grid area transactions.

\subsection{Stylized Facts of Price Dynamics in Intraday Energy Markets}\label{sec:stylized-facts}

Modeling the dynamics of intraday energy markets requires a nuanced understanding of the limit order book (LOB), particularly at high frequencies where granular trading behavior dominates. The LOB aggregates market orders, limit orders, and cancellations, which collectively shape price movements. Intraday energy markets, being inherently high-frequency, necessitate models that explicitly capture these LOB dynamics to accurately reflect trading conditions.

A distinctive feature of intraday markets is the variation in LOB activity throughout the trading session. Order flow and price changes intensify as delivery approaches, driven by increased urgency to rebalance and improved forecast accuracy. Figure~\ref{fig:midprices} illustrates mid-price fluctuations for hourly delivery products at 18:00, 19:00, and 20:00 on September 30, 2023, showing higher volatility near the beginning and end of the trading sessions.

\begin{figure}
\centering 
\includegraphics[width=\linewidth]{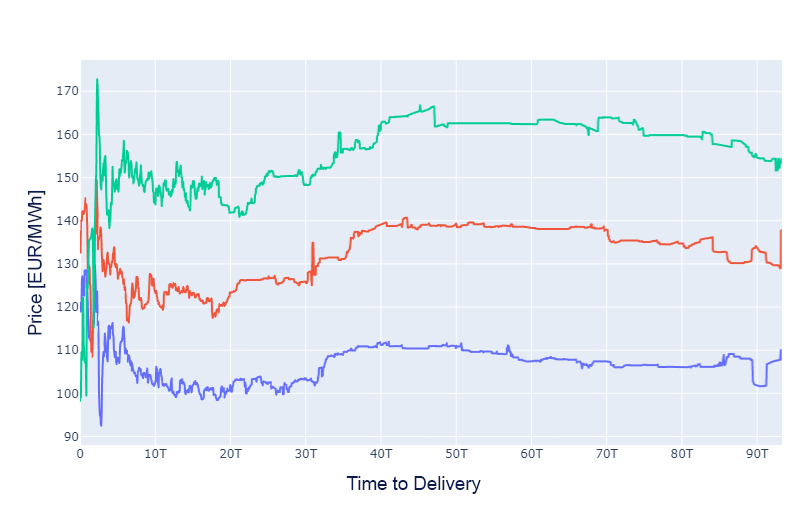} 
\caption{German transaction prices for trading sessions on September 30, 2023, for hourly delivery products at 18:00, 19:00, and 20:00, relative to time to gate closure.}
\label{fig:midprices} 
\end{figure}

Figure~\ref{fig:VarIntensity1} highlights the empirical intensity of buy (see Figure \ref{fig: EmpiricalBaseline} in Section~\ref{sec:complementaryRes} for the sell side) market orders throughout the trading day, which peaks during the final hours before delivery. The midday activity increase, which is highlighted in the empirical intensity of buy orders (see Figure \ref{fig: EmpiricalBaseline} for the sell side of the order book), corresponds to higher liquidity during noon delivery hours (products with a delivery 12:00-15:00), likely driven by Germany's rapid expansion of solar energy capacity~\cite{Ember} as well as the general expected increase of trading activity mid-day. Conversely, activity diminishes sharply during the final trading hour due to market decoupling, which restricts cross-border trading and limits transactions to within Germany. These patterns underscore the importance of modeling the time-varying nature of LOB activity.

Figure~\ref{fig:SignaturePlot3d1} displays \emph{signature plots} of realized volatility for all $24$ German hourly intraday products.  
For a given terminal horizon $T$ and sampling interval $\Delta$ (in seconds), we compute
$$
C_T(\Delta) = \sum_{k=1}^{\lfloor T/\Delta \rfloor}\!\big(S_{t_k}-S_{t_{k-1}}\big)^2,
\qquad t_k = k\,\Delta,
$$
where $S_{t_k}$ denotes the mid-price at time $t_{k}$. The figure overlays four surfaces, corresponding to $T \in \{8,7,6,5\}\,\mathrm{h}$, mapping $\Delta \mapsto C_T(\Delta)$ across the hourly trading products. Two main patterns emerge:  
Firstly, the \emph{Samuelson effect}, i.e. for any fixed $\Delta$, $C_T(\Delta)$ increases as time to delivery shortens (from $T=8$\,h to $T=5$\,h), reflecting higher intraday volatility closer to gate closure in line with~\cite{Deschatre2023,DG22,RePEc:arx:papers:2211.13002}. Secondly, \emph{microstructure noise}, as for each product and horizon $T$, $C_T(\Delta)$ decreases as $\Delta$ grows over the $0$–$300$\,second range. This pattern is consistent with noise-induced inflation of high-frequency quadratic variation and mean reversion at very short scales. If prices followed a pure Brownian motion without microstructure effects, the signature plot would remain flat in $\Delta$, showing no dependence of volatility on the sampling frequency.

These empirical observations highlight the need for models that capture both the rise in volatility as gate closure approaches and the intra-session clustering of activity.  
Hawkes processes provide a natural framework for this task: by allowing the intensity to depend on its own history, they reproduce the self-exciting and clustering behavior characteristic of high-frequency financial data, features that Brownian motion or simpler point processes such as Poisson models cannot adequately represent. In particular, Hawkes dynamics account for the strong autocorrelation observed in mid-price returns and order arrivals in intraday energy markets. Another reason to favor Hawkes processes over, e.g., time inhomogeneous Poisson processes, is that time changed arrival times of market orders cannot be adequately modeled using exponential distributions~\cite{DG22}. Furthermore, we observe that the integrated volatility of the mid-price process varies over the trading horizon, increasing as gate closure nears (see Figure~\ref{fig:SignaturePlot3d1}). Again, this phenomenon reflects the heightened trading intensity and market activity driven by increasingly accurate renewable energy forecasts as delivery approaches. As traders receive more reliable information closer to delivery time, they adjust their positions accordingly, resulting in increased volatility.

\begin{figure*}[!htbp]
  \centering
  \begin{subfigure}[b]{0.49\textwidth}
    \centering
    \includegraphics[width=\linewidth,trim={8pt 14pt 8pt 8pt},clip]{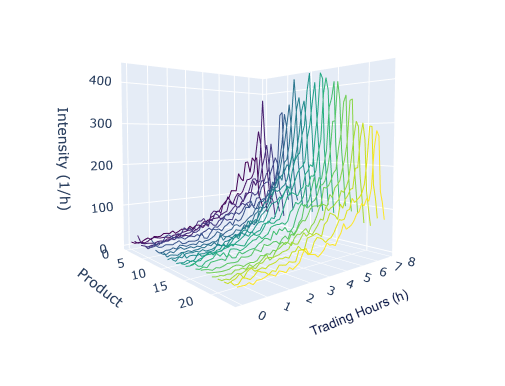}
    \subcaption{Empirical intensity of buy market orders (1/h) over trading \emph{time}.}
    \label{fig:VarIntensity1}
  \end{subfigure}
  \hfill
  \begin{subfigure}[b]{0.45\textwidth}
    \centering
    \includegraphics[width=\linewidth,trim={6pt 10pt 6pt 6pt},clip]{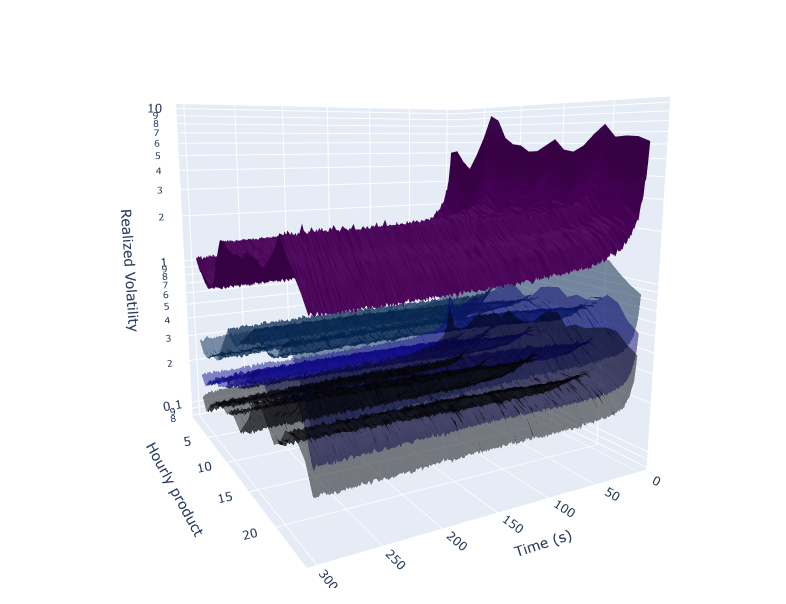}
    \subcaption{Signature-plot surfaces $C_T(\Delta)$ for the 24 hourly products.}
    \label{fig:SignaturePlot3d1}
  \end{subfigure}

  \caption{(a) Empirical buy-side intensity over the trading day.
  (b) Realized-volatility signature surfaces across sampling scales and delivery hours; layered horizons $T\in\{8,7,6,5\}$\,h with x-axis the sampling interval $\Delta$ (s); y-axis the delivery hour/products from 1–24 and z-axis shows $C_T(\Delta)$ on $[0,T]$. 
    Layers correspond to $T\in\{8,7,6,5\}$\,h (counted back from gate closure). 
    $C_T(\Delta)$ decreases in $\Delta$ (indicating microstructure noise) and, at fixed $\Delta$, rises near delivery (showing the Samuelson effect).}
  \label{fig:combined-intensity-signature}
\end{figure*}

 In Section~\ref{sec:model-optimal-exec}, we provide a concise introduction to the  marked Hawkes price model used to describe the market order. Moreover, in Section~\ref{sec:ModelCalibration}, we demonstrate that the Hawkes process framework, with time-varying baseline intensities and event clustering through the self/cross excitation kernels, provides a robust approach to addressing these dynamics. Specifically, we calibrate the Hawkes model to German intraday energy market data using maximum likelihood estimation (MLE), ensuring an accurate fit to the observed market behavior, see Tables~\ref{tab:combinedParams} and~\ref{tab:ks-test-results} below for results on the calibrated parameter space and Section~\ref{sec:GoodnessofFitResults} for a comprehensive study on the goodness-of-fit analysis of the different models. 

Moving beyond traditional parametric approaches, we employ a spline-based method with fixed knots to model baseline intensity changes over time. This non-parametric approach offers the flexibility to capture complex intensity patterns, including the exponential increase in activity as gate closure approaches. Additionally, it accommodates the sharp decline in activity observed during the final trading hour, following market decoupling~\cite{Bal22}, a phenomenon often overlooked in existing studies.

\subsection{Optimal Execution in Intraday Energy Markets} 
Beyond modeling order flow, it is essential to evaluate how these dynamics influence trading strategies. For energy companies managing substantial trading volumes, optimal timing and execution are critical to minimizing both market impact and trading costs. By leveraging Hawkes processes to model intraday dynamics, traders can anticipate periods of high liquidity or volatility, enabling them to adjust their strategies more effectively. While much of the literature on optimal execution focuses on equity markets and price impact models~\cite{almgren2000optimal, Gatheral2013, Obizhaeva2005, Bertsimas1998, NV22}, research on optimal execution in energy markets remains limited. Existing studies~\cite{KZ20, AGP16, glas2020intraday, JCP24} often overlook the granular effects of limit order book (LOB) events on execution prices. Furthermore, most works assume that prices evolve as continuous martingales under passive trading, incorporating all order impacts. However, at high frequencies, large orders exert substantial effects that invalidate these assumptions~\cite{Alfonsi2016,AlfonsiExtension}.

To overcome these limitations, we employ a high frequent, tick-by-tick execution price model based on Hawkes processes, which effectively captures the influence of market orders from other participants in intraday energy markets. High-frequency models~\cite{Bacry2013, BTJF16, LN19} are well-suited for capturing key statistical properties such as autocorrelation in trade signs, volatility clustering and microstructure noise. While these approaches have primarily been applied to equity markets, recent extensions have begun to explore energy markets~\cite{DG22, Deschatre2023}, incorporating marked Hawkes processes to model mid-price jump sizes. These developments lay the foundation for more nuanced and effective optimal execution strategies tailored to the unique characteristics of energy markets.

This paper contributes to the literature on optimal execution in intraday energy markets by combining high frequency price models~\cite{DG22, Deschatre2023} with optimal execution frameworks. Hawkes processes in our model effectively capture the flow of market orders, while price impact models are used to calculate trading costs. Additionally, our approach incorporates time-dependent baseline intensities to reflect the sharp increase in market activity near gate closure~\cite{DESCHATRE2021105504, KIESEL201777, Luckner} and market decoupling thereafter. The transient market impact component, which accounts for the influence of past trades on prices, has only recently been introduced in electricity markets~\cite{JCP24} with focus on energy arbitrage and battery energy storage system optimization problems. To the best of our knowledge, no existing work has implemented a fit of the model and an optimal trading strategy directly to real-market data, making this study a novel contribution to the field. Our backtesting analysis in Section~\ref{sec:optimal-strateg-backtesting} addresses the following key questions: 
\begin{itemize} 
\item What are the transaction costs of the renewable portfolio compared to naive execution strategies such as Time-Weighted Average Price (TWAP) and Volume-Weighted Average Price (VWAP)? 
\item How much transaction cost savings does the optimized strategy provide over TWAP and VWAP? 
\item Should the optimal strategy be applied universally across all hourly products, or selectively, to maximize cost efficiency? 
\end{itemize}

\subsection{Structure of the Paper}

This paper is structured as follows. Section~\ref{sec:model-optimal-exec} introduces the Hawkes price model and formulates the optimal execution problem. It presents the derivation of the optimal execution strategy using Hawkes processes to model buy and sell market orders, as detailed in Theorem~\ref{thm:optimalstrategy}, allowing for a quantitative analysis of execution costs. Section~\ref{sec:ModelCalibration} details the calibration protocol for the model using EPEX Spot LOB data. The Hawkes parameters and its baseline intensity are estimated independently for hourly German intraday market products, and the goodness-of-fit is evaluated with comparisons between models under varying assumptions about event dependencies and intensities. Section~\ref{sec:optimal-strateg-backtesting} presents the back-testing analysis, assessing the optimal execution strategy against TWAP and VWAP benchmarks using real market data from 2023–2024. The appendices include additional results on maximum likelihood estimation (Section~\ref{sec:MLEHawkes}), goodness-of-fit analysis (Section~\ref{sec:GoodnessofFit}), and some complementary findings (Section~\ref{sec:complementaryRes}).

\section{The Hawkes Price Model and Optimal Execution Strategies}

\label{sec:model-optimal-exec}

In this section, we provide an introduction to the Hawkes price model used to capture tick-by-tick price data in intraday energy markets and present a few important properties of the model. Thereafter, we formulate the optimal execution problem for this model.

\subsection{The Marked Hawkes Price Model}

We introduce the Marked Hawkes Price (MHP) model using Poisson random measures and counting processes. Consider $M$ maturities $0 < T_1 < T_2 < \dots < T_M = T$, where $M \in \mathbb{N}$. Let $K \subseteq \mathbb{R}_+$ denote the set of possible marks 
(e.g., price jumps or trade sizes), and let $\mathcal{K}$ be the Borel sigma-algebra on $K$. On a probability space $(\Omega, \mathcal{F}, \mathbb{P})$, we consider $2M$ independent marked point processes representing upward and downward price jumps for each maturity or equivalently, the arrival of buy and sell market orders.

For each maturity $T_m$, where $m = 1, \dots, M$, we define the counting processes $N_m^+(t)$ and $N_m^-(t)$, which count the number of buy and sell market orders up to time $t$, respectively, and set $N_m(t) := \big(N_m^{+}(t),N_m^{-}(t) \big)^{\top}$. The counting processes can be represented via Poisson random measures $\pi_m^+(dt, dy)$ and $\pi_m^-(dt, dy)$ on the measurable space $\left( E, \mathcal{E} \right) = \left( \mathbb{R}_+ \times K, \mathcal{B}\left( \mathbb{R}_+ \right) \otimes \mathcal{K} \right)$, each with intensity measure $\lambda_{m,t}^\pm \, dt \otimes \nu(dy)$, where $\nu$ is a probability measure on $(K, \mathcal{K})$ satisfying $\nu(\{0\}) = 0$. As before we set $\pi_m(dt,dy) = \left(\pi_m^+(dt,dy),\pi_m^-(dt,dy) \right)$. The cumulative upward and downward price jump processes for maturity $T_m$ are given by:

\begin{equation}
    \label{eq:pricemodel}
    S_{m,t}^\pm = \int_0^t \int_{K} y \, \pi^{\pm}_m(ds, dy) = \sum_{k=1}^{N^{\pm}_m(t)} Y_{m,k}^\pm,
\end{equation}

where $Y_{m,k}^\pm$ are the marks associated with the $k$-th upward or downward jump, respectively. The net price process is thus given by

\begin{align}
    \label{eq:netprice}
    S_{m,t} &= S_{m,0} + S_{m,t}^+ - S_{m,t}^- \\
    &= S_{m,0} + \int_0^t \int_{K} y \left( \pi_m^+(ds, dy) - \pi_m^-(ds, dy) \right),\nonumber
\end{align}
where $S_{m,0}$ is the initial price for the product with maturity $T_m$. Note that we do not impose positivity of $\Price_{m,t}$ and prices may be negative, as regularly observed in intraday electricity markets. The set of marks $\MarkSpace\subseteq\Rp$ refers to \emph{non-negative magnitudes} (e.g., absolute tick sizes), while the sign of price changes is carried by the buy/sell decomposition via $N_m^+(t)$ and $N_m^-(t)$.

To capture the clustering of order arrivals, characteristic of Hawkes processes, we define the conditional intensities of the counting processes $N_m(t)$ as functions of the history of past events. More precisely, for $t \geq 0$, we define the vector of conditional intensities $ \lambda_{m,t} = \big( \lambda^+_{m,t},\lambda^-_{m,t} \big)^{\top}$ as  
\begin{equation}
    \label{eq:hawkes_intensity}
    \lambda_{m,t}
    =
    \bar{\lambda}_m(t) 
    +  \int_0^{t} \int_K 
      \phi_{m}^\pm(t-s,y)\,\pi_m(ds,dy),
\end{equation}
where for each hourly product $m = 1,\dots, M$ we denote by
\begin{itemize}
  \item $\bar{\lambda}_m(t) := \big(\lambda_{m,\infty}^{+}(t), \lambda_{m,\infty}^{-}(t)\big)^{\top}$ the deterministic baseline intensity,  
  \item $\phi_{m}^\pm(\cdot,y)$ the positive definite kernel matrix describing the self- and cross-exciting effects between buy and sell market orders. 
\end{itemize}
Thus, each past event of type $(\pm)$ (buy/sell) with mark $y$ increases the intensity $\lambda_{m,t}$ at future times $t>s$ according to the kernel $\phi_{m}^\pm$. The conditional intensities therefore capture dependence on past events through the convolution of the kernels with the counting processes, thereby reflecting the self-exciting and, if present, mutually exciting nature of price jumps. To illustrate the mechanism in a simple setting, Figure~\ref{fig:HawkesExample} shows a simulation of a univariate Hawkes process with exponential decay kernel, where each arrival temporarily increases the intensity, which then decays back toward its baseline. Note that under standard Hawkes stability condition (spectral radius of the kernel matrix $<1$) and finite baseline intensity, the number of jumps on any finite horizon is almost surely finite~\cite{Bacry2013}.

\begin{figure}
    \centering
    \begin{subfigure}[b]{0.4\textwidth}
    \centering
    \includegraphics[width=\textwidth]{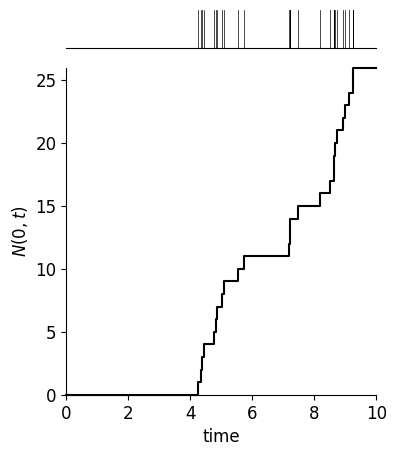}
    \caption{Counting process $N_t$}
    \end{subfigure}
\hfill
    \begin{subfigure}[b]{0.4\textwidth}
    \centering
    \includegraphics[width=\textwidth]{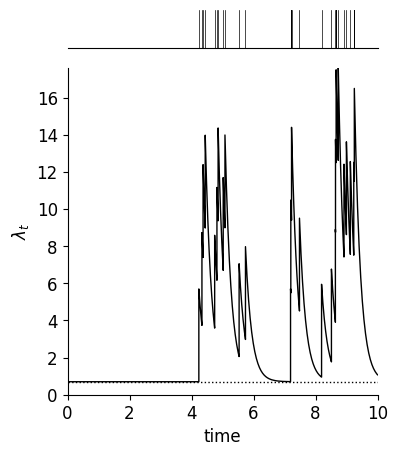}
    \caption{Intensity process $\lambda_t$}
    \end{subfigure}
    \caption{Sample paths of the intensity process and the associated counting process $N_t$ with $\alpha = 1$, $\beta = 5$, and $\lambda_{\infty} = 0.7$. Sample paths of a univariate Hawkes process $N_t$ with intensity $ \lambda_t = \lambda_{\infty} + \int_{0}^{t} \alpha \E^{-\beta(t - s)} \, dN_s$ for $\alpha = 1$, $\beta = 5$, and $\lambda_0 = 0.7$.}
    \label{fig:HawkesExample}
\end{figure}

Note that the family of counting processes $\{N_m(t)\}_{m=1}^M$ forms a $2M$-dimensional Hawkes process.  
In our setting, however, we abstract from cross-product interactions and do not model correlations across different delivery products. Hence, no cross-excitation effects are included between products.
For ease of notation, we therefore fix a single delivery product and drop the index $m$. Buy and sell arrivals are then modeled by a bivariate Hawkes process with exponential decay kernels. Specifically, we let 
$$
\boldsymbol{\lambda}_t := \begin{pmatrix}\lambda_t^+ \\ \lambda_t^-\end{pmatrix}, 
\qquad 
\bar{\boldsymbol{\lambda}}_t := \begin{pmatrix}\lambda_\infty^+(t) \\ \lambda_\infty^-(t)\end{pmatrix},
$$
denote the conditional and baseline intensities (with the latter possibly time-varying).  
We assume that the kernel matrix is separable, of the form

$$
\phi^{\pm}(u,y) = \phi(u) \odot \Phi(y),
$$
where $\odot$ denotes the Hadamard product, $\varphi(u)$ captures the exponential decay in time, and $\Phi(y)$ encodes the mark dependence via

\begin{equation}
\label{eq:impactmatrix}
\Phi(y) = 
\begin{pmatrix}
\varphi^{(s)}(y) & \varphi^{(c)}(y) \\
\varphi^{(c)}(y) & \varphi^{(s)}(y)
\end{pmatrix}, 
\quad \varphi^{(s)}(y), \varphi^{(c)}(y) \geq 0.
\end{equation}
Here, $\varphi^{(s)}$ and $\varphi^{(c)}: K \to \mathbb{R}_+$ denote the mark-dependent self- and cross-excitation amplitudes, where the superscripts $(s)$ and $(c)$ refer to self- and cross-excitation, respectively. In the special case where $\varphi^{(s)}$ and $\varphi^{(c)}$ are constant functions, the model reduces to the standard bivariate Hawkes process.
Throughout, we assume the following moment condition to hold

\begin{align}\label{eq:firstmoments1}
i_{\cdot} := \int_K \varphi^{(\cdot)}\!\left(\frac{y}{m_1}\right)\nu(dy) < \infty,
\end{align}
where $m_1 = \int_K y\,\nu(dy)$ denotes the expected mark size. In the remainder, we standardize the mark sizes by $m_1$ and let the conditional intensity vector be given by
\begin{equation}
\label{eq:intensityprocess}
\boldsymbol{\lambda}_t 
= \bar{\boldsymbol{\lambda}}_t 
+ \int_0^t \phi(t-u)\,\odot\,\Phi
\begin{pmatrix}
    \tfrac{dN^{+}(u)}{m_1} \\
    \tfrac{dN^{-}(u)}{m_1}
\end{pmatrix},
\end{equation}

where the integral is understood component-wise and we note that $\int_0^t \varphi^{(s)}\left(\tfrac{dN^{+}(u)}{m_1}\right) = \int_0^t \int_K \varphi^{(s)} \left( y/m_1\right)\pi^+\left( du,dy\right)$.

 The self-excitation function $\varphi^{(s)}$ captures the effect of past events of the same type (buy or sell) on the current intensity of that type, while the cross-excitation function $\varphi^{(c)}$ captures the effect of past events of the opposite type on the current intensity. The conditions in~\eqref{eq:firstmoments1} ensure that the expected contributions from the excitation functions are finite, which is necessary for the process to be well-defined. 

To model the volatility components of the price model, we further assume that the second moments of the excitation functions are finite, i.e.
\begin{align*}
    \int_K \left[ \varphi^{(\cdot)}\left( \dfrac{y}{m_1} \right) \right]^2 \nu(dy) < \infty.
\end{align*}

Assuming the shape of intraday seasonality is known, we define the \emph{deseasonalized} intensity process as  
\begin{equation}\label{eq:intensityprocessdeseasonal}
   \tilde{\boldsymbol{\lambda}}_t := \boldsymbol{\lambda}_t - \bar{\boldsymbol{\lambda}}_t, 
   \qquad t \geq 0,
\end{equation}
where $\tilde{\boldsymbol{\lambda}}_t=(\tilde{\lambda}^+_t,\tilde{\lambda}^-_t)^\top$ represents the residual intensity obtained by removing the baseline $\bar{\boldsymbol{\lambda}}_t$.

We recall the empirical observation of increasing intensity and volatility as time to maturity decreases, known as the Samuelson effect (see Figure~\ref{fig:VarIntensity1}). Incorporating a time-varying baseline intensity in~\eqref{eq:intensityprocess} effectively captures this phenomenon. While most studies use a parametric approach, typically assuming an exponential function for the baseline intensity, we propose a more flexible non-parametric approach using spline functions, such as cosine bump functions, see Section~\ref{sec:estimation-hawkes} below.

\subsection{Optimal Execution}
\label{sec:OptimalExecution}

In this section, we describe the framework of the optimal execution problem in our model class. Throughout this section, we consider a single asset, specifically an hourly product in the intraday energy market, and denote its price at time $t$ by $P_t$. We assume that $P_t$ can be decomposed into a fundamental price component $S_t$ given by~\eqref{eq:netprice} and a transient price deviation $D_t$, such that: 
\begin{align}\label{eq:executionprice}
    P_t = S_t + D_t.
\end{align}

We define $N_t$ as the cumulative signed volume of past market orders \emph{excluding the trader's own trades} up to time $t$, given by $$N_t := N_t^+ - N_t^-,$$ where $N_t^+$ and $N_t^-$ are the cumulative volumes of buy and sell orders from other market participants, respectively. By convention, a buy order increases $N_t$ and $S_t$, while a sell order decreases them. 

Under the assumption that an order modifies the price proportionally to its size, and operating within the framework of a block-shaped limit order book, we assume that a part $\mu \in [0,1]$ of the price impact is permanent and the remainder $1 - \mu$ is transient, decaying exponentially at rate $\rho > 0$. Accordingly, before considering the trader's own trading activity, we model the dynamics of $S_t$ and $D_t$ as:
\begin{align}\label{eq:allprices}
    dS_t &= \mu \, dN_t, \\
    dD_t &= - \rho D_t \, dt + (1 - \mu) \, dN_t.
\end{align}

Our goal is to specify how the trader's trading activity modifies the price and to determine the cost induced by their trading. We consider a trader aiming to liquidate an initial position $x_0 > 0$ over a given time interval $[0, T]$. Let $X_t$ represent the trader's cumulative trading volume up to time $t$ (with $X_0 = 0$ and $X_T = x_0$). We assume that the trader's trading activity affects the processes $S_t$ and $D_t$ as follows:
\begin{align}
    dS_t &= \mu \, dN_t + \epsilon \, dX_t, \\
    dD_t &= - \rho D_t \, dt + (1 - \mu) \, dN_t + (1 - \epsilon) \, dX_t,
\end{align}
where $\epsilon \in [0,1]$ represents the proportion of the trader's price impact that is permanent. For simplicity, and since we do not impose additional assumptions on the information structure of other market participants,  we set $\epsilon = \mu$, i.e., the trader's impact is assumed to be identical to that of the other participants.

Analogously to some existing literature (see \cite{Alfonsi2012, Obizhaeva2005}; for intraday markets, see~\cite{KZ20}), we will assume throughout a block-shaped limit order book. When the trader executes an order of size $\nu \in \mathbb{R}$ (with $\nu > 0$ for a buy order and $\nu < 0$ for a sell order), they incur a cost:
\begin{align} \label{eq:blockshaped}
    \pi_t(\nu) = \int_0^{\nu} \left( P_t + y \right) dy = P_t \nu + \frac{1}{2} \nu^2,
\end{align}
since the post-trade price is $P_{t+} = P_t + \nu$. This cost corresponds to trading all the volume at the average price $\frac{1}{2}(P_t + P_{t+})$.

Let $D_X$ be the (at most countable) set of jump times of $X$ and $\Delta X_{\tau}:= X_{\tau_+} - X_{\tau}$. Summing the block costs over trades before terminal time $T$ and adding the terminal block needed to liquidate any residual inventory $X_T$ at time $T$ yields the following liquidation cost
\begin{align}\label{eq:costfunction-ab}
C(X) 
&= \int_{[0,T)} P_u\, dX_u 
 + \frac{1}{2}\sum_{\tau\in D_X\cap[0,T)} (\Delta X_\tau)^2\nonumber\\
 &\quad - P_T X_T + \frac{1}{2} X_T^2.\\
 & = \int_{0}^{T}\! P_u \, dX_u +\! \frac{1}{2} \int_{0}^{T}\! d[X]_u \!- P_T X_T \!+ \frac{1}{2} X_T^2.\nonumber
\end{align}
For a simple strategy with jumps $\{\tau_i<T\}$, the total paid is 
$\sum_i \big(P_{\tau_i}\Delta X_{\tau_i} + \tfrac{1}{2}(\Delta X_{\tau_i})^2\big)$ which follows by \eqref{eq:blockshaped} and block-shaped assumption. 
If $X_T\neq 0$, liquidation at $T$ adds the terminal execution cost $P_T(-X_T)+\tfrac{1}{2}X_T^2$. 
Passing to the c\`adl\`ag limit gives \eqref{eq:costfunction-ab}. The Stieltjes integral $\int_{[0,T)} P_u\,dX_u$ already accounts for the $P_{\tau_i}\Delta X_{\tau_i}$ terms. Hence by denoting with $[X]_t$ the quadratic variation of $X$ up to time $t$ and letting $X_t = X_t^c + \sum_{\tau\in D_X\cap[0,T)} \left( X_{\tau}^+ -  X_{\tau}\right)$ the last equality in \eqref{eq:costfunction-ab} follows.

\begin{remark}
As we can see from the definition of the market order arrival rates given by equation~\eqref{eq:intensityprocessdeseasonal}, the trader's orders (i.e., our trading activity) do not impact the jump rates $\lambda^+$ and $\lambda^-$, since there is no additional $dX_t$ term influencing the arrival rates of subsequent market orders. This contrasts with the market orders issued by other traders (the market activity which we calibrate directly from the LOB). The main reason for this modeling choice is tractability and the assumption that each individual trader tends to send several orders of the same sign in a row, which in turn creates autocorrelation in the signs of trades~\cite{Alfonsi2016} and this effect is stronger than the mutual excitation between different traders. However, recent work by Horst et al.~\cite{CHTH23} accounts for the fact that the trading activity of a large market participant, whose trades may impact the market, triggers child orders and endogenously affects the future order flow.
\end{remark}

In this setting, we study the optimal execution problem under the assumption that buy and sell market order flows follow Hawkes processes. The arrival times of buy and sell orders are modeled as the jump times of two dependent Hawkes processes, $N^+$ (buy) and $N^-$ (sell), with conditional instantaneous intensities defined by~\eqref{eq:intensityprocessdeseasonal}. The dependence between these processes is captured by the impact matrix in Equation \eqref{eq:impactmatrix}.

We now derive the optimal execution strategy in the most general setting.  
In the empirical results and backtesting section, we will study two different modeling approaches.  
A key feature of the Marked Hawkes price model is that the optimal execution problem can be solved in closed form: the linear price impact and the exponential decay kernels of the Hawkes processes ensure Markovian dynamics.  
More generally, Markovian representations can also be recovered for completely monotone decay kernels, as shown by Alfonsi et al.~\cite{Alfonsi2012}.  
For clarity throughout this section, we work in the \emph{deseasonalized} case assuming the same constant baseline intensities for buy and sell side.

Before introducing the optimal trading strategy and the main theorem, we define auxiliary variables that simplify the dynamics.  
Instead of working directly with $\lambda_t^+$ and $\lambda_t^-$, we consider
\[
  \kappa_t = \lambda_t^+ - \lambda_t^-, 
  \qquad 
  \gamma_t = \lambda_t^+ + \lambda_t^-,
\]
where $\kappa_t$ describes the \emph{order-flow imbalance} and $\gamma_t$ the \emph{total order intensity}.  
By the variation-of-constant formula, these processes satisfy the following stochastic differential equations:
\begin{align*}
d\kappa_t &= -\beta \kappa_t\, dt + dI_t,\\
d\gamma_t &= -\beta \big( \gamma_t - \lambda_{\infty} \big) dt + d\bar{I}_t,
\end{align*}
where $\lambda_{\infty} = \lambda_{\infty}^+ + \lambda_{\infty}^-$ is the sum of baseline intensities with $\lambda_{\infty}^+=\lambda_{\infty}^-$, and the processes $I_t$ and $\bar{I}_t$ encode the impact of arrivals and are given by
\begin{align}\label{eq:intensitiesvol2}
I_t &=\! \int_0^t\! \big( (\varphi^{(s)} - \varphi^{(c)}) \tfrac{dN_s^+}{m_1} - (\varphi^{(s)} - \varphi^{(c)}) \tfrac{dN_s^-}{m_1} \big),\\
\bar{I}_t &=\! \int_0^t\! \big( (\varphi^{(s)} + \varphi^{(c)}) \tfrac{dN_s^+}{m_1} + (\varphi^{(s)} + \varphi^{(c)}) \tfrac{dN_s^-}{m_1} \big).
\end{align}

Let $\left( \tau_i \right)_{i \in \mathbb{N}}$ denote the ordered random jump times of the counting process $N$, with $\tau_0 = 0$. For $t \in [0, T]$, let $\chi_t$ denote the total number of jumps of $I$ that have occurred between time $0$ and $t$. Then, the process $\kappa_t$ can be expressed as:

\begin{equation} \label{eq:delta_t}
\kappa_t = \kappa_0 \E^{-\beta t} + \sum_{\ell = 1}^{\chi_t} \E^{-\beta (t - \tau_{\ell})} \Delta I_{\tau_{\ell}} = \E^{-\beta t} \left( \kappa_0 + \Theta_{\chi_t} \right),
\end{equation}

where $\Delta I_{\tau_{\ell}}$ denotes the jump of $I$ at time $\tau_{\ell}$, and we define:

\begin{equation} \label{eq:Theta_i}
\Theta_i = \sum_{j = 1}^{i} \E^{\beta \tau_j} \Delta I_{\tau_j}, \quad i \geq 1,
\end{equation}

with $\Theta_0 = 0$. This representation shows that for $i \geq 0$ and $t \in [\tau_i, \tau_{i+1})$, the quantity $\E^{\beta t} \kappa_t = \kappa_0 + \Theta_i$ depends only on the number of jumps up to time $t$, i.e., $\chi_t = i$.

The jump intensity of $(N_t)_{t\geq 0}$ is determined by the Markov processes $\kappa_t$ and $\gamma_t$ as defined in Equations~\eqref{eq:intensitiesvol2} taking values in $\mathbb{R}$ and $\mathbb{R}_{+}$, respectively. 
 The state variables of the optimization problem are then $\left( X_t, D_t, S_t, \kappa_t, \gamma_t \right)$, and the control is the trading strategy $\left( X_t \right)_{t \in [0, T]}$, with $X_0 = x_0$.

Our objective is to minimize the expected cost associated with the trading strategy $X$. The cost functional on a generic time interval $[t, T]$ associated with the liquidation strategy $X$ is given by

\begin{align}\label{eq:expected_cost} C(t, X) &= \int_{t}^{T} P_u dX_u + \dfrac{1}{2} \int_{t}^{T} d[X, X]_u - P_T X_T \nonumber\\
&\quad + \dfrac{1}{2} X_T^2. 
\end{align}

Let $\mathcal{A}_t$ denote the set of admissible strategies on $[t, T]$. The value function of the optimal control problem is then defined as:

\begin{equation}
\label{eq:value_function}
\mathcal{C}(t, Z_t) 
= \inf_{X \in \mathcal{A}_t} 
\mathbb{E}\left[ C(t, X) \,\big|\,Z_t=(x,d,z,\kappa,\gamma)\right],
\end{equation}
for $Z_t=(X_t, D_t, S_t, \kappa_t, \gamma_t)$.
The terminal condition for the value function is:

\begin{align}\label{eq:terminal_condition} 
\mathcal{C}(T, x, d, z, \kappa, \gamma) = - (d + z) x + \dfrac{1}{2} x^2. 
\end{align}

The optimal strategy $X^*$ can be derived analytically in this model due to the Markovian and affine structure of the state space and the fact that the cost function (objective) is a linear-quadratic function of the state variables~\cite{AlfonsiExtension}. For this, one defines a continuously differentiable function $\mathcal{C}(t, x, d, z, \kappa, \gamma)$ a priori along with an admissible strategy $X^*$, and define the process $\Pi_t(X)$ as
\begin{align}
\label{eq:mainprocess}
   \!\! \Pi_t(X) = \!\int_0^t\! P_u\, dX_u \!+\! \frac{1}{2} \int_0^t \!d[X, X]_u + \mathcal{C}(t,Z_t).
\end{align}

Finally, one has to verify that $\Pi_t(X)$ is a submartingale for any admissible strategy $X$, and $\Pi_t(X^*)$ is a martingale for $X^*$ the solution of the optimal control problem~\eqref{eq:value_function}. This optimality criterion gives then rise to the following theorem. 

\begin{theorem}\label{thm:optimalstrategy}
    Let $\tilde{\alpha} = i_s - i_c$ and $\eta = \beta - \tilde{\alpha}$ and define the two continuously differentiable functions $\zeta, \omega: \mathbb{R} \to \mathbb{R}$ as
\begin{equation}
\label{eq:zeta_omega_def}
\begin{alignedat}{2}
\zeta(y) &= \begin{cases}
1, & y = 0, \\
\frac{1 - \E^{-y}}{y}, & y \neq 0,
\end{cases}
\qquad
\omega(y) &= \begin{cases}
\frac{1}{2}, & y = 0, \\
\frac{\E^{-y} - 1 + y}{y^2}, & y \neq 0.
\end{cases}
\end{alignedat}
\end{equation}
Let $h := T - t$. Then, the optimal strategy satisfies
 \begin{align}
 \label{OptimalStrategyFull1}
         X_t^* &= - (1 - \mu)^{-1}[1 + \rho h] D_t + \frac{m_1}{2 \rho} [2 + \rho h] \\
        &\quad \times \kappa_t \left[ 1 + \frac{\rho h}{2 + \rho h} \left( \zeta\left( h \eta \right) + \mu \rho h\, \omega\left( h \eta \right) \right) \right]\nonumber,
    \end{align}
where $\kappa_t = \lambda_t^+ - \lambda_t^-$ is the intensity order flow imbalance as before.
\end{theorem}
\begin{proof}
This follows from~\cite[Theorem 4.1]{Alfonsi2016} and \cite[Theorem 2.2]{AlfonsiExtension} under the assumption of mono-exponential kernels for the Hawkes intensity and the propagator kernel.
\end{proof}

\section{Model Calibration and Backtesting}\label{sec:ModelCalibration}

This section outlines the methodology for calibrating the model presented in Section~\ref{sec:model-optimal-exec}. The price model in Equation~\eqref{eq:allprices} consists of two components, allowing for an independent calibration of each part. We focus on the estimation of the Hawkes process parameters. Finally, we relate the market impact parameters to the microstructure of the limit order book (LOB) and demonstrate how these can be estimated using publicly available EPEX Spot market data. 

\subsection{Estimation of the Hawkes Parameters}\label{sec:estimation-hawkes}

We estimate the Hawkes process parameters using maximum likelihood estimation (MLE), which is established in the literature (e.g., \cite{Ozaki1979MaximumLE}). The log-likelihood for multivariate Hawkes processes, along with its gradient and Hessian, is detailed in Appendix~\ref{sec:MLEHawkes}. Numerical optimization is performed using a quasi-Newton method.

To capture the increasing intensity of buy and sell market order (MO) arrivals as gate closure approaches, as observed in Section~\ref{sec:stylized-facts}, we model the baseline intensity in~\eqref{eq:intensityprocess} as time-dependent. Assuming all market orders have unit volume simplifies the impact kernel in~\eqref{eq:impactmatrix} to a unit matrix.

For the baseline intensity, we employ a non-parametric approach using cubic B-splines and cosine bump functions. B-splines are particularly effective for modeling time-varying intensity patterns, such as the U-shape observed in equity markets~\cite{Chen_Hall_2013, cartea2015algorithmic}. To our knowledge, this non-parametric framework has not been applied to intraday energy markets. The baseline intensity is defined as:
$$
\lambda_{\infty}^{\pm}(t) = \exp\left( \sum_{i=1}^{N_{\text{basis}}} \xi_i^{\pm} f_i(t) \right),  
$$

where $ \{f_i(t)\} $ are the basis functions, and $ \{\xi_i^{\pm}\} \in \mathbb{R} $ are parameters to be estimated.

The basis functions are constructed using cosine bump functions defined as:

$$
\hat{f}(x) =
\begin{cases}
    \frac{\cos\left(\frac{\pi x}{2}\right) + 1}{4}, & \text{if } |x| \leq 2, \\
    0, & \text{otherwise}.
\end{cases}
$$

These are then transformed as follows:

$$
f_i(t) = \hat{f}\left( \frac{t - T_{b,m} - (i - 2)w}{w} \right), 
$$
where $ i = 1, 2, \dots, N_{\text{basis}} $, $ w = \frac{l}{N_{\text{basis}} - 3} $, and $ l = T_{e,m} - T_{b,m} $. Here, $ N_{\text{basis}} $ represents the number of basis functions, which can be determined using a grid search. This construction provides the flexibility to reproduce the empirical shape of market activity given by an exponential rise toward gate closure followed by a sharp decline, while retaining smoothness and interpretability. While increasing $ N_{\text{basis}} $ improves the model’s flexibility, it also increases computational complexity by enlarging the parameter space, making the non-convex optimization problem more challenging.

We compare the spline-based approach with a piecewise constant approximation, evaluating both self-excitation-only models and those incorporating cross-excitation. Using tick-by-tick EPEX Spot data for hourly German market products, the spline approximation consistently outperforms the piecewise method in terms of goodness-of-fit metrics. Interarrival time analyses for compensated marginal and pooled processes further confirm the effectiveness of the spline approach for modeling buy/sell market order arrivals. Model selection using Akaike Information Criterion (AIC) and likelihood differences reinforces these results. Detailed results are provided in Appendix~\ref{sec:GoodnessofFitResults}.

\begin{figure}
    \centering
    \includegraphics[width=0.9\linewidth]{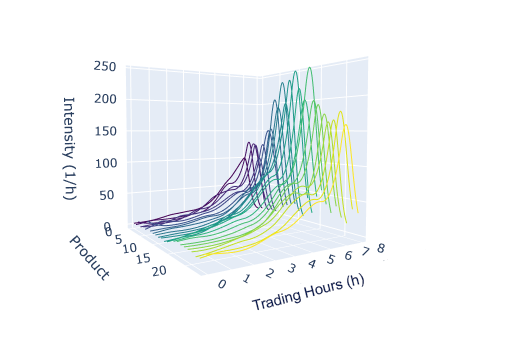}
    \caption{Calibrated baseline intensity of buy market orders ($1/h$) using the spline approximation method. The plot shows the median baseline intensity for each day and hourly product, using median calibrated parameters obtained via MLE.
    }
    \label{fig: CalibratedSpline}
\end{figure}

On average, the shape of the baseline function for each delivery hour closely aligns with the empirical observations. This level of accuracy would be difficult to achieve using a monotone increasing baseline function (e.g., exponential form) or the piecewise approximation methodology described earlier (see Figure~\ref{fig:PieceWiseSpline} in the Appendix).

Furthermore, Figure~\ref{fig: CalibratedSpline} highlights that the median calibrated spline functions capture the slight seasonality across trading products. Specifically, the intensities exhibit an increase from the early hourly products, peak around the mid-day products, and slightly decrease for the afternoon products. These patterns align with the empirical analysis in Section~\ref{sec:stylized-facts}, particularly the observations illustrated in Figure~\ref{fig:VarIntensity1} for the buy-side order book (LOB).

We focus on the calibration results obtained using the spline approximation model. Figure~\ref{fig:SplineApproximationComparison} presents the results for hourly products with delivery between 10:00 and 12:00. The number of basis functions used for the baseline intensity estimation is $N_{\text{basis}} = 10$.

\begin{figure*}[!htbp]
    \centering
    \includegraphics[width=0.9\linewidth]{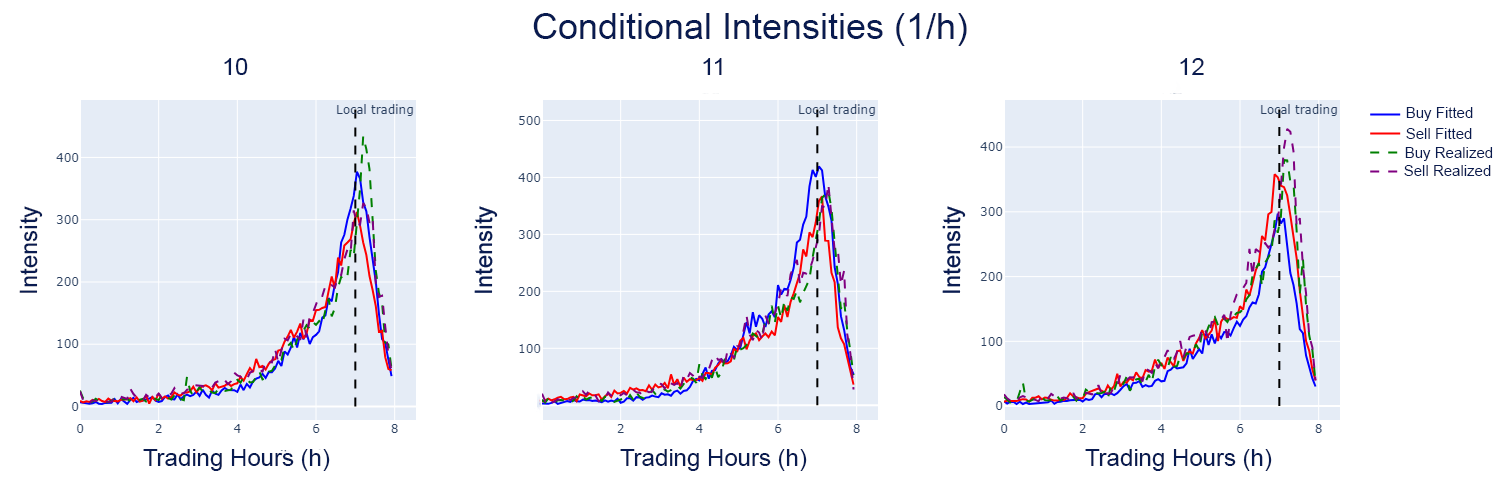}
    \caption{Empirical and fitted intensity for Buy/Sell Market order. We estimate the empirical intensity by counting the average number of events during $[t,t+ \delta t]$, where $N(t)$ the realized counting process for each side of the LOB. 
    }
    \label{fig:SplineApproximationComparison}
\end{figure*}

Next, we analyze the parameters of self-excitation exponential kernels using the spline approximation for the baseline intensity. Table~\ref{tab:combinedParams} presents the median estimated parameters for each delivery hour, separately for the Buy and Sell sides of the limit order book (LOB). The median is employed as a robust estimator to reduce the impact of outliers and achieve a balance between bias and variance.

On the buy side, the median excitation parameter $\alpha$ across the hourly trading products ranges from $0.116$ to $0.171\,\persec$, corresponding to approximately $418$–$616\,\perhour$ after conversion.  
This implies that following the arrival of a Buy market order, the conditional intensity increases by about $418$ to $616$ arrivals per hour, with a median around $473\,\perhour$.  
Similar values are observed on the Sell side, with only minor variations across products.  
The decay rate $\beta$ on the Buy side corresponds to half-lives $\ln(2)/\beta$ between $1.2$ and $2.2$ seconds, indicating that the self-excitation effect decays rapidly.  
On the Sell side, the corresponding half-lives range from $1.3$ to $2.1$ seconds.  
The branching ratio $\alpha/\beta$, which in the univariate exponential Hawkes model represents the expected number of additional arrivals triggered by one event, lies between $0.28$ and $0.36$ for most products.  
This satisfies the stability condition $\alpha/\beta < 1$ required for exponential Hawkes processes~\cite{Bacry2013}.  
Economically, this means that a single Buy market order leads on average to an additional $0.28$–$0.36$ Buy orders.  
Sell-side estimates again display similar ranges.  
Kolmogorov–Smirnov tests reveal no significant differences at the 5\% level between Buy and Sell sides in the distributions of excitation parameters and branching ratios.  
We also find no systematic differences across delivery hours in the Hawkes kernel parameters.  
By contrast, the calibrated baseline intensity obtained from the spline approximation exhibits pronounced variation across products (Figure~\ref{fig: CalibratedSpline}), in line with the intraday activity patterns reported in Figure~\ref{fig:VarIntensity1}.

\begin{table}[!htbp]
  \centering
   \resizebox{\linewidth}{!}{
  \begin{tabular}{l|rrrr|rrrr}
    \toprule
    {\textbf{Hourly product}} & \multicolumn{4}{c|}{\textbf{Buy Side}} & \multicolumn{4}{c}{\textbf{Sell Side}} \\
    \cmidrule(lr){2-5} \cmidrule(lr){6-9}
    & $\alpha$ ($\frac{1}{s}$) & $\beta$ ($\frac{1}{s}$) & $\frac{\alpha}{\beta}$ & $\frac{\ln(2)}{\beta}$ (sec.) & $\alpha$ ($\frac{1}{s}$) & $\beta$ ($\frac{1}{s}$) & $\frac{\alpha}{\beta}$ & $\frac{\ln(2)}{\beta}$ (sec.) \\
    \midrule
    1 & 0.123 $\pm$ 0.066 & 0.345 $\pm$ 0.256 & 0.357 & 2.009 & 0.150 $\pm$ 0.090 & 0.430 $\pm$ 0.326 & 0.349 & 1.612 \\
    2 & 0.117 $\pm$ 0.057 & 0.329 $\pm$ 0.214 & 0.356 & 2.107 & 0.136 $\pm$ 0.067 & 0.412 $\pm$ 0.252 & 0.330 & 1.682 \\
    3 & 0.119 $\pm$ 0.056 & 0.371 $\pm$ 0.273 & 0.321 & 1.868 & 0.130 $\pm$ 0.064 & 0.407 $\pm$ 0.276 & 0.319 & 1.703 \\
    4 & 0.118 $\pm$ 0.057 & 0.352 $\pm$ 0.216 & 0.335 & 1.969 & 0.117 $\pm$ 0.064 & 0.378 $\pm$ 0.290 & 0.310 & 1.834 \\
    5 & 0.124 $\pm$ 0.057 & 0.363 $\pm$ 0.201 & 0.342 & 1.909 & 0.134 $\pm$ 0.075 & 0.421 $\pm$ 0.337 & 0.318 & 1.646 \\
    6 & 0.125 $\pm$ 0.063 & 0.375 $\pm$ 0.267 & 0.333 & 1.848 & 0.133 $\pm$ 0.070 & 0.438 $\pm$ 0.331 & 0.304 & 1.583 \\
    7 & 0.171 $\pm$ 0.087 & 0.529 $\pm$ 0.370 & 0.323 & 1.310 & 0.160 $\pm$ 0.083 & 0.520 $\pm$ 0.391 & 0.308 & 1.333 \\
    8 & 0.160 $\pm$ 0.093 & 0.494 $\pm$ 0.391 & 0.324 & 1.403 & 0.141 $\pm$ 0.074 & 0.417 $\pm$ 0.295 & 0.338 & 1.662 \\
    9 & 0.168 $\pm$ 0.088 & 0.584 $\pm$ 0.490 & 0.288 & 1.187 & 0.144 $\pm$ 0.074 & 0.444 $\pm$ 0.345 & 0.324 & 1.561 \\
    10 & 0.153 $\pm$ 0.071 & 0.455 $\pm$ 0.298 & 0.336 & 1.523 & 0.144 $\pm$ 0.065 & 0.445 $\pm$ 0.294 & 0.324 & 1.558 \\
    11 & 0.137 $\pm$ 0.065 & 0.395 $\pm$ 0.241 & 0.347 & 1.755 & 0.149 $\pm$ 0.074 & 0.418 $\pm$ 0.259 & 0.356 & 1.658 \\
    12 & 0.116 $\pm$ 0.053 & 0.317 $\pm$ 0.171 & 0.366 & 2.187 & 0.136 $\pm$ 0.062 & 0.379 $\pm$ 0.231 & 0.359 & 1.829 \\
    13 & 0.116 $\pm$ 0.051 & 0.315 $\pm$ 0.183 & 0.368 & 2.200 & 0.132 $\pm$ 0.062 & 0.385 $\pm$ 0.223 & 0.343 & 1.800 \\
    14 & 0.121 $\pm$ 0.055 & 0.337 $\pm$ 0.197 & 0.359 & 2.057 & 0.122 $\pm$ 0.052 & 0.365 $\pm$ 0.209 & 0.334 & 1.899 \\
    15 & 0.116 $\pm$ 0.045 & 0.318 $\pm$ 0.157 & 0.365 & 2.180 & 0.119 $\pm$ 0.046 & 0.341 $\pm$ 0.161 & 0.349 & 2.033 \\
    16 & 0.128 $\pm$ 0.057 & 0.381 $\pm$ 0.232 & 0.336 & 1.819 & 0.131 $\pm$ 0.054 & 0.373 $\pm$ 0.183 & 0.351 & 1.858 \\
    17 & 0.123 $\pm$ 0.051 & 0.356 $\pm$ 0.193 & 0.346 & 1.947 & 0.128 $\pm$ 0.057 & 0.383 $\pm$ 0.224 & 0.334 & 1.810 \\
    18 & 0.122 $\pm$ 0.054 & 0.381 $\pm$ 0.267 & 0.320 & 1.819 & 0.121 $\pm$ 0.058 & 0.339 $\pm$ 0.214 & 0.357 & 2.045 \\
    19 & 0.133 $\pm$ 0.065 & 0.394 $\pm$ 0.260 & 0.338 & 1.759 & 0.124 $\pm$ 0.056 & 0.345 $\pm$ 0.189 & 0.359 & 2.009 \\
    20 & 0.140 $\pm$ 0.072 & 0.420 $\pm$ 0.303 & 0.333 & 1.650 & 0.112 $\pm$ 0.046 & 0.323 $\pm$ 0.184 & 0.347 & 2.146 \\
    21 & 0.134 $\pm$ 0.075 & 0.427 $\pm$ 0.353 & 0.314 & 1.623 & 0.120 $\pm$ 0.059 & 0.368 $\pm$ 0.269 & 0.326 & 1.884 \\
    22 & 0.138 $\pm$ 0.073 & 0.406 $\pm$ 0.289 & 0.340 & 1.707 & 0.126 $\pm$ 0.056 & 0.403 $\pm$ 0.267 & 0.313 & 1.720 \\
    23 & 0.143 $\pm$ 0.089 & 0.462 $\pm$ 0.393 & 0.310 & 1.500 & 0.132 $\pm$ 0.062 & 0.401 $\pm$ 0.259 & 0.329 & 1.729 \\
    24 & 0.121 $\pm$ 0.062 & 0.358 $\pm$ 0.240 & 0.338 & 1.936 & 0.124 $\pm$ 0.058 & 0.402 $\pm$ 0.257 & 0.308 & 1.724 \\
    \bottomrule
  \end{tabular}}
  \caption{
Medians ($\pm$ standard deviations) of the estimated parameters of the exponential Hawkes kernel 
$t \mapsto \alpha e^{-\beta t}$ for Buy and Sell market orders, conditional on delivery hour. 
Reported are $\alpha$ and $\beta$ (in $\persec$), the branching ratio $\alpha/\beta$, and the half-life $\ln(2)/\beta$ (in seconds).  
For intuition, $\alpha=0.12\,\persec$ corresponds to an instantaneous increase of about $+432\,\perhour$ in the conditional intensity after one buy arrival, 
while $\beta=0.40\,\persec$ implies a half-life of $\ln 2/0.40 \approx 1.73$ seconds.}
  \label{tab:combinedParams}
\end{table}

We evaluate the spline model by comparing its inter-event time distribution directly to the realized distribution. This is achieved by simulating a process of the same length as the empirical process using the thinning algorithm, with the estimated parameters provided in Table~\ref{tab:combinedParams}. The simulated point process is then compared to the empirical one using inter-event time distributions, visualized in a QQ-plot. Additionally, we perform a two-sided Kolmogorov-Smirnov (KS) test to verify whether the simulated and empirical inter-event time distributions are statistically indistinguishable. For this test, we simulate event times for Buy/Sell market orders for each day and each hourly product.

Table~\ref{tab:ks-test-results} summarizes the mean p-values obtained from the KS test and the percentage of cases where the null hypothesis of equal distributions is not rejected at the significance level of 5\%. The results indicate that the null hypothesis of equal distributions is not rejected at $\alpha = 0.05$ for each hourly product, demonstrating a good overall fit of the model to the empirical data.

Figures~\ref{fig:qq-plot1} and~\ref{fig:qq-plot2} present QQ-plots of inter-event times for all events on the bid side of the order book, focusing on two selected maturities (19h and 21h). For short inter-event times ($\Delta\tau \leq 300s$), the model provides a good fit. However, the fit deteriorates for longer inter-event times, particularly for the shorter maturity (19h). This discrepancy may be attributed to lower trading activity at the start of the trading session. In some sessions, this inactivity significantly extends the tail of the inter-event time distribution, making it slightly left-skewed. This behavior is not fully captured by the calibrated model, potentially limiting its ability to accurately simulate the distribution of long inter-event times.



\begin{figure*}[!htbp]
    \centering
    \begin{subfigure}[b]{0.45\textwidth}
    \centering
    \includegraphics[width=\linewidth]{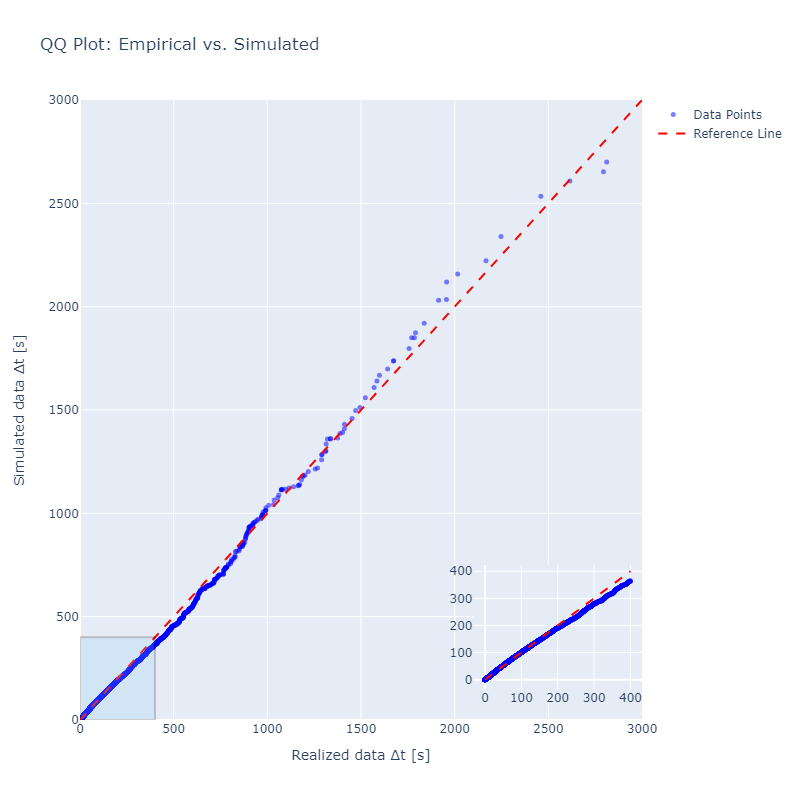}
    \caption{QQ-plot for Buy MOs (Simulated vs Realized) for maturity 21h}
    \label{fig:qq-plot1}
    \end{subfigure}
\hfill
    \begin{subfigure}[b]{0.45\textwidth}
    \centering
    \includegraphics[width=\linewidth]{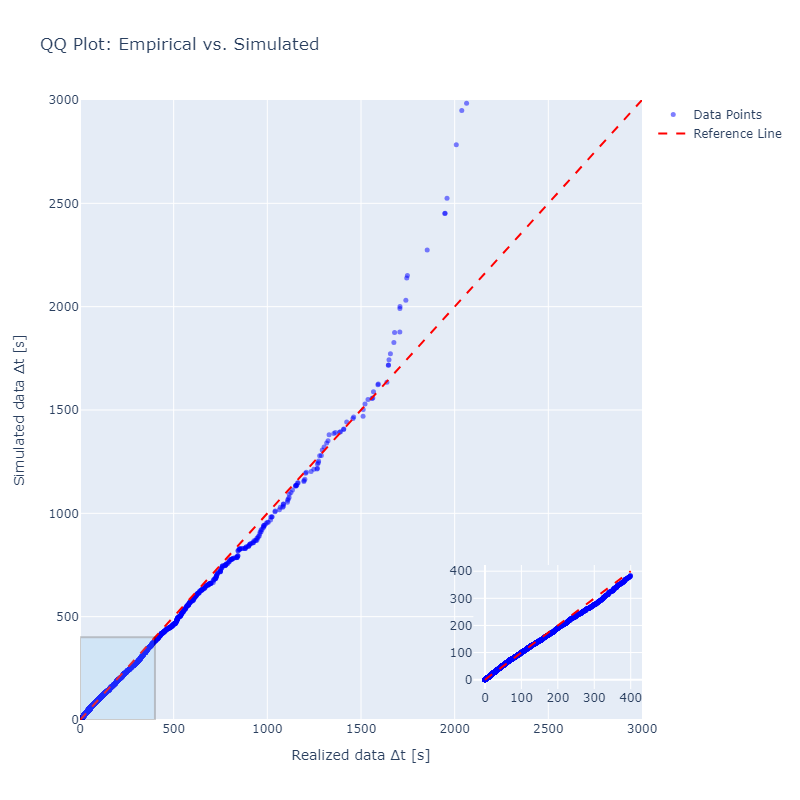}
    \caption{QQ-plot for Buy MOs (Simulated vs Realized) for maturity 19h}
    \label{fig:qq-plot2}
    \end{subfigure}
\end{figure*}\noindent

\begin{table}[!htbp]
    \centering
    \scriptsize
    \setlength{\tabcolsep}{4pt}
    \renewcommand{\arraystretch}{1.1}
    \begin{tabular}{>{\centering\arraybackslash}p{1cm}|>{\centering\arraybackslash}p{3cm}|>{\centering\arraybackslash}p{3cm}}
        \toprule
        \textbf{Hour} & \textbf{Mean ± Std P-Value} & \textbf{Percentage $\geq$ 0.05} \\
        \midrule
        0  & 0.22 $\pm$ 0.08 & 69.37\% \\ 
        1  & 0.24 $\pm$ 0.07 & 70.27\% \\ 
        2  & 0.31 $\pm$ 0.06 & 73.08\% \\ 
        3  & 0.29 $\pm$ 0.05 & 79.83\% \\ 
        4  & 0.28 $\pm$ 0.05 & 79.67\% \\ 
        5  & 0.27 $\pm$ 0.07 & 80.17\% \\ 
        6  & 0.23 $\pm$ 0.09 & 76.52\% \\ 
        7  & 0.23 $\pm$ 0.09 & 77.48\% \\ 
        8  & 0.24 $\pm$ 0.08 & 69.53\% \\ 
        9  & 0.21 $\pm$ 0.08 & 66.39\% \\ 
        10 & 0.22 $\pm$ 0.08 & 69.42\% \\ 
        11 & 0.21 $\pm$ 0.09 & 66.40\% \\ 
        12 & 0.20 $\pm$ 0.10 & 61.06\% \\ 
        13 & 0.21 $\pm$ 0.09 & 68.42\% \\ 
        14 & 0.25 $\pm$ 0.08 & 68.64\% \\ 
        15 & 0.25 $\pm$ 0.07 & 76.19\% \\ 
        16 & 0.26 $\pm$ 0.07 & 76.92\% \\ 
        17 & 0.23 $\pm$ 0.08 & 73.28\% \\ 
        18 & 0.22 $\pm$ 0.08 & 70.31\% \\ 
        19 & 0.22 $\pm$ 0.08 & 71.97\% \\ 
        20 & 0.21 $\pm$ 0.09 & 67.97\% \\ 
        21 & 0.20 $\pm$ 0.10 & 72.36\% \\ 
        22 & 0.23 $\pm$ 0.09 & 71.77\% \\
        23 & 0.27 $\pm$ 0.07 & 77.60\% \\ 
        \bottomrule
    \end{tabular}
    \caption{Kolmogorov-Smirnov test p-values between the simulated interarrival times and the realized ones, along with the percentage of p-values $\geq 0.05$ for each delivery hour.}
    \label{tab:ks-test-results}
\end{table}

So far we have considered only the case where the cross-covariances are equal to zero and whilst the counting process $\left(N_t^+,N_t^-\right)$ describing the arrival of Buy/Sell market orders respectively may be self-exciting, they are independent and thus not mutually exciting. To introduce the excitation between the two types of orders we consider the parameterization of the conditional intensity function given by the full impact matrix kernel \eqref{eq:impactmatrix}. Following the goodness-of-fit analysis we have conducted (See Appendix \ref{sec:GoodnessofFitResults}) in the univariate case we can see that overall the fitted model with the spline approximation of the baseline intensity provides a better fit to the data. Hence as a baseline intensity, we will consider the one given by the spline approximation. We present the calibrated parameters for the bivariate case in the Appendix below \ref{tab:combined-table}. Our goal is to test the performance of the two models to cost savings compared to the benchmark models; this will be considered in the following section.

In order to calculate the resilience speed of the transient part of the price we will use the approximation derived in~\cite{Chen2019}. The authors derived the resilience speed parameter for a TWAP type of schedule, assuming that the market orderflow follows a self-exciting Hawkes process with exponential kernel $t \mapsto \alpha \E^{-\beta t}$. Under this setting they showed that a good approximation of the resilience speed is $\rho \approx \frac{1 - \alpha/\beta}{T/2}$, where $[0,T]$ the time period where the trader want to liquidate his position. The approximation shows how the resilience parameter links to the self-exciting dynamics of market order flow. As the authors mention this resilience speed compromises only an approximation since different liquidating strategies would in principle lead to different resilience speed parameters $\rho$. To this end we follow this approach and we leave out the calibration of a propagator type of model for future research.

\subsection{Execution Costs}
\label{ExecutionCostEstimates}
We will now present a summary of the cost estimates for the different liquidation strategies we are going to propose in the following section. In order to perform such a transaction cost analysis we should first find some appropriate cost estimates. Since we do not have any endogenous data, we proceed to estimate various market
impact factors from public data disseminated by the exchange and the LOB \cite{EPEXSpotSE}. Our calibration methodology is similar in nature to the analysis presented in \cite{cartea2015algorithmic,glas2020intraday} where an approximation of the half-spread and instantaneous trading costs is studied using LOB market data. 

\subsubsection{Instantaneous impact estimation from stochastic order book}
We consider an LOB model with a total depth of $K$ price levels. The bid- and ask-side liquidity are denoted by $L_i^{b}$ and $L_i^{a}$, respectively, representing the amount of standing liquidity at the $i$-th price level away from the best bid or ask. Rather than imposing an equally spaced grid of price level, which may introduce gaps of zero liquidity ($L_i^{b} = 0$ for some levels $i$), we build an LOB representation restricted to positive standing liquidity at each price level, ensuring a non-empty book for estimation. 

Let $M_0 := L_0^b$, with $M_0 > 0$, and define $\Delta P_0 := 0$. Let 
\[
K_0 := \#\bigl\{ i \in \{0,\dots,K\} : L_i^b > 0 \bigr\}.
\]
We then construct a sequence of price offsets $\Delta P_i$ and standing liquidity $M_i$ recursively for  $i = 1,\dots,K_0$:
\[
\Delta P_i := \inf\{j > \Delta P_{i-1} : L_j^b > 0\}, 
\quad M_i := L_{\Delta P_i}^b.
\]
Hence, the book offset $\Delta P_i$ denotes the $i$-th non-empty price level, and $M_i$ is the corresponding liquidity. The cumulative liquidity up to price level $\Delta P_i$ is
\[
m_i := \sum_{j=0}^i M_j.
\]
Thus, the LOB state at any given time can be fully described by the random vector $\{(M_i, \Delta P_i)\}_{i=0,\dots,K_0}$. Compared with simple snapshot averages, a non-parametric estimation of the joint distribution $\{(M_i, \Delta P_i)\}$ can capture more nuanced stochastic behavior, including asymmetries and heavier tails, leading to better estimates of transaction costs.

When building a deterministic grid of price offsets, we replace random offsets with their empirical means, still denoted by $\Delta P_i$. In this framework, the distribution of $M_i$ alone suffices to compute the transaction cost function
\begin{equation}
\label{eq:transaction_cost}
\begin{aligned}
C(\mathcal{U}) 
:= \frac{1}{\mathcal{U}}\, \mathbb{E}\bigg[
&\sum_{j=1}^K \mathbf{1}_{\{m_{j-1}<\mathcal{U} \leq m_j\}} \Big(
\sum_{i=1}^j M_{i-1}\,\Delta P_{i-1} \\
&+ (\mathcal{U} - m_j)\,\Delta P_j \Big)
\bigg],
\end{aligned}
\end{equation}

where $\mathcal{U}$ is the size of a large market order. In general, estimating such a high-dimensional joint distribution is prone to the curse of dimensionality, requiring a balance between bias and variance. To alleviate complexity, a modified cost function can be used to reduce dimensionality while still offering a meaningful estimate of the instantaneous impact. 

We define an approximate cost function for instantaneous market impact as follows:
\begin{align}\label{eq:marketimpact}
    \hat{C}(\mathcal{U})
    &:=\!\frac{1}{\mathcal{U}} \sum_{j=1}^K 
    \mathbb{E}\Bigl[
        M_{j-1}\Delta P_{j-1} \!+\! (\mathcal{U} -\! m_j)\Delta P_j 
        \Big|V_{i,j}
\Bigr]\nonumber\\
    &\qquad \times \mathbb{P}\big(m_j < \mathcal{U} \le m_{j+1}\big),
\end{align}

where $V_{i,j}=\{\omega\in\Omega\colon m_j < \mathcal{U} \le m_{j+1}\}$ and $\nu$ is the market order size, $\Delta P_j$ are the price offsets, and $M_j$ denotes the liquidity at offset $\Delta P_j$. The index $m_j$ is the aggregate volume up to level $j$.

Although \eqref{eq:marketimpact} is simpler to estimate than the exact cost function, it introduces a bias. In particular, once an order consumes liquidity up to level~$i$, the effective price is approximated by $\Delta P_i$ instead of a volume-weighted average 
$\sum_{j=0}^i M_j \Delta P_j(\sum_{j=0}^i M_j)^{-1}$.
However, this approximation significantly reduces the computational burden and estimation variance, since it only requires estimating the joint distribution at two levels $(m_i, m_{i+1})$ each time.

To approximate the joint density of the random variables $x = m_i$ and $y = m_{i+1}$, we use a bivariate kernel density estimator. Suppose we have $\bigl\{(x_t, y_t)\bigr\}_{t=1}^n$, where $n$ denotes the number of available data points, drawn from the distribution of $(m_i, m_{i+1})$. The estimator is:
\begin{equation}
    \hat{f}_{x,y}(x,y)
    = \frac{1}{n} \sum_{t=1}^n K_H\bigl(x - x_t,\, y - y_t\bigr),
\end{equation}
where $K_H$ is a bivariate kernel with bandwidth matrix $H$. In our analysis, we use a Gaussian kernel, choosing $H$ via standard bandwidth selection techniques.

To estimate the instantaneous market impact, we shall simulate the order book in its stationary state, followed by the submission of orders of size $\mathcal{U}$. For each order, we compute the per-unit cost (price change per MWh) $\hat{C}(\mathcal{U})$ from~\eqref{eq:marketimpact}. Next, we fit a linear relationship:
$$
    \hat{C}(\mathcal{U}) = \eta\,\mathcal{U} + \varepsilon,
    \quad \varepsilon \sim \mathcal{N}\bigl(0, \sigma^2\bigr),
$$
where $\hat{C}(\mathcal{U})$ is the price change per MWh after executing an order of size $\mathcal{U}$ MWh. We interpret $\eta$ (price/MWh$^2$) as the instantaneous market impact factor. 

In practice, we calibrate the distribution of liquidity using kernel density estimation on quote snapshots $\{t_i\}_{i=1}^n$, each taken within a one-hour ($T=1\text{h}$) bin of the limit order book. We use data from the last eight tradable hours for each hourly product, yielding a total of $N_{\mathrm{days}} \times N_{\mathrm{products}} \times 8$ parameters for the instantaneous scaling factor~$\hat{\eta}$. We estimate~$\hat{\eta}$ via ordinary least squares, focusing primarily on the sell side of the LOB, as our interest lies in liquidating a long position. 

Figure~\ref{fig:TempImpactCost} shows the average piecewise-interpolated function describing the linear execution cost coefficient. The median calibrated $\hat{\eta}(t)$ for each hourly contract clearly varies over time; specifically, $\hat{\eta}(t)$ decreases after the trading session begins (measured relative to the last eight trading hours) and slightly increases in the final hour (local trading regime). This aligns with the notion that liquidity diminishes near gate closure, causing higher transaction costs~\cite{Bal22, glas2020intraday, KZ20}.

\begin{figure}
    \centering
    \includegraphics[width=\linewidth]{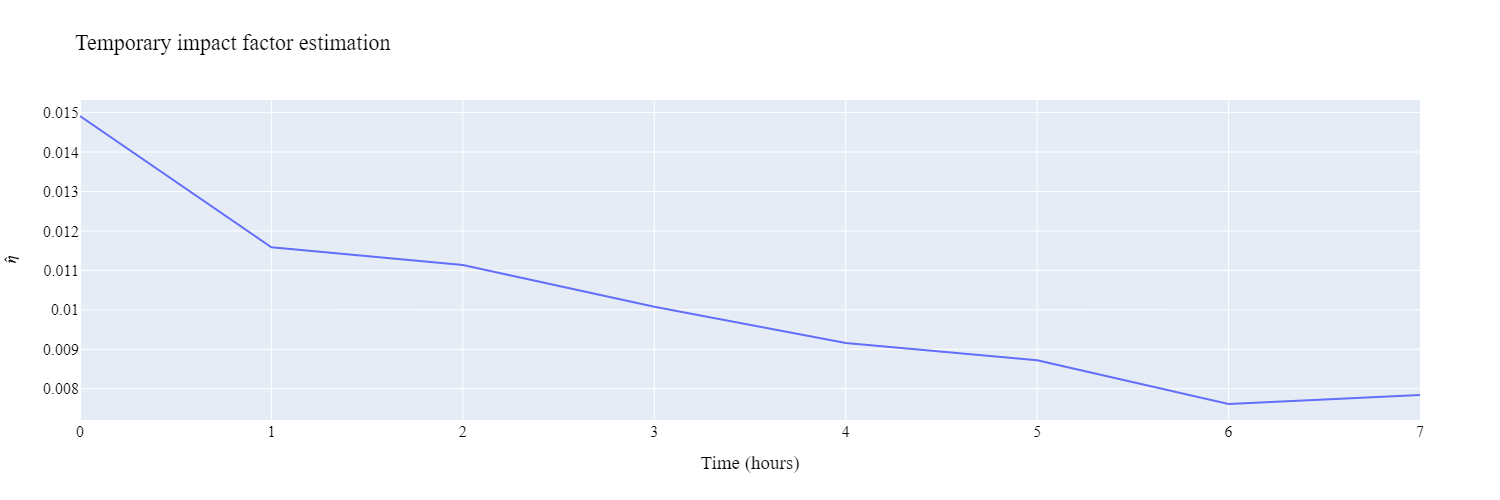}
    \caption{Estimated \emph{temporary impact} coefficients with respect to gate closure for the sell side of the LOB. 
    These impact coefficients are obtained by simulating a large market order that ``walks'' through the limit order book, assuming a linear relationship 
    \(\hat{C}(\mathcal{U}) = \eta\,\mathcal{U} + \varepsilon\), where \(\varepsilon \sim \mathcal{N}(0,\sigma_\varepsilon^2)\) and \(\hat{C}(\mathcal{U})\) denotes the price change per MWh for an order of size \(\mathcal{U}\).}
    \label{fig:TempImpactCost}
\end{figure}

\subsubsection{Bid-Ask Spread Transaction Costs}

Following \cite{Bal22,glas2020intraday,KZ20}, the bid-ask spread (BAS) in intraday energy markets exhibits a clear seasonal pattern. In particular, the average BAS decreases significantly at the outset of the trading session and remains nearly constant for the final eight hours of trading. One hour before maturity, when trading is restricted to the local market, the spread widens sharply. This behavior mirrors the pattern observed for the half-spread, which closely tracks the shape of the temporary impact coefficient $\hat{\eta}(t)$. Both observations suggest that liquidity is notably more limited at market open and after market decoupling compared to just before decoupling \cite{glas2020intraday,KZ20}.

Consequently, the \emph{cost of immediacy}, i.e. submitting a market order instead of a limit order, tends to be substantially higher at the start of trading and during the final trading hour. In contrast, during the intermediate trading period, these costs remain relatively lower (see Figure~\ref{fig:TempImpactCost}). 

To backtest the strategies using half-spread execution costs, we compute a time-weighted bid-ask spread (BAS) over discrete intervals of 6 seconds. Consider an interval $I_i := (T_{i-1}, T_i] \subset [0, T]$, and suppose the BAS changes $N_i$ times within this interval at timestamps
\begin{align*}
T_{i-1} < t_1^{(i)} < t_2^{(i)} < \dots < t_{N_i}^{(i)} \leq T_i.    
\end{align*}

Let $\text{BAS}_j$ denote the bid-ask spread on the sub-interval 
$(t_j^{(i)}, t_{j+1}^{(i)})$, for $j = 0, \dots, N_i$, where we set 
$t_0^{(i)} := T_{i-1}$ and $t_{N_i+1}^{(i)} := T_i$. 
The time-weighted bid-ask spread on the interval $I_i$ is then defined as
\begin{equation}
\label{eq:timeWeightedBAS}
\overline{BAS}_i := \frac{1}{\Delta T_i} \sum_{j=0}^{N_i} 
BAS_j \cdot \left( \min(t_{j+1}^{(i)}, T_i) - \max(t_j^{(i)}, T_{i-1}) \right).
\end{equation}

\section{Optimal strategy backtesting and transaction cost analysis}\label{sec:optimal-strateg-backtesting}
In this section, we apply our model, along with the associated closed-form solution of the optimal liquidation strategy to a transaction cost analysis framework.

\subsection{Scaling and discretization of the optimal strategy}

We start by describing the use of the optimal strategy described in Equation \eqref{OptimalStrategyFull1}, which reads

\begin{align}
\label{OptimalStrategyFull2}
&X_t^* = -\frac{1+\rho h}{\,1-\epsilon\,}\,D_t^* \\[0.4em]
&+ \frac{m_1(2+\rho h)}{2(1-\epsilon)\rho}
\Biggl(\!1 +\! \frac{\rho h}{2+\rho h}
    \Bigl[\zeta(\eta h) \!+ \mu \rho h \,\omega(\eta h)\Bigr]\Biggr)\kappa_t.\nonumber
\end{align}

We recall that $\kappa_t = \lambda_t^+ - \lambda_t^-$, is the vector of intensity imbalances and are modeled with one of the aforementioned models using the seasonal versions of the baseline intensity, which has been estimated non-parametrically as described in Section \ref{sec:ModelCalibration}. We are then going to consider the univariate and bivariate models separately for the backtesting and transaction cost analysis study using the calibrated parameters obtained (See Table \ref{tab:combinedParams} for the univariate case and Table \ref{tab:combined-table} for the bivariate case). In order to use the optimal liquidating strategy we consider a subset $\Theta$ of the time interval  $\Theta \subset[T_{b,m},T_{e,m}]\subset [0,T]$ possible made of stopping times and trade for each $t \in \Theta$ the quantity 

\begin{align}
\label{OptimalStrategyProxy}
\xi_{t,T}^s & = - \frac{[1+\rho h]qs D_t + X_t}{2+ \rho h}\\
& + \frac{m_1}{2\rho}\left( 1+ \frac{\rho h}{2+ \rho h}\times \left[\zeta h \eta) + \mu \rho \omega( h\eta)\right]\cdot s \kappa_t\right),\nonumber
\end{align}
so that equation \eqref{OptimalStrategyFull2} holds for $t+\Delta t$ if $s=1$. In order to calculate the resilience part of the execution price we consider the simplified expression
\begin{equation}
    D_t^* = \sum_{\tau \leq t} \Delta N_{\tau}\left[ G(t - \tau) - G(\infty)\right].
\end{equation}

The propagator kernel is chosen as in the OW framework $G(u) = (1-\mu)\exp(-\rho u) + \mu$ and $\Delta N_{\tau}$ the jumps of the order-flow. As we have already described we have a positive (negative) unit jump when we have an arrival of a buy (sell) market order. We also introduce the scaling factor $s\in [0,1]$ that multiplies both $\kappa_t$ and $D_t$ and is a hyper-parameter that tunes the leveraging effect of the strategy, describing the deviation of the whole strategy from the standard liquidation scheme in~\cite{Obizhaeva2005}. We can observe that in the case $s\!=\!0$ we derive the static strategy, which is nothing else by the block trades that have been derived in the OW strategy, and in the case $s\!=\!1$ we derive the strategy given by optimal execution strategy we derived under the Marked Hawkes process which drives the order flow and presented in Theorem \ref{thm:optimalstrategy}.

\subsection{Benchmark strategies}

In optimal execution, instant order book execution involves trading the full volume at once in a single block trade, represented mathematically as $dX_t = - X_0 \delta_0(dt)$. This approach, equivalent to a single click on the bid or ask side in the EPEX system, sacrifices cost efficiency for instant execution. However, due to higher market impact and execution costs, especially when trading far from maturity, it is excluded from our backtesting methodology. A very common strategy used in the context of optimal execution is the so-called time-weighted average price or TWAP. The strategy $X^{TWAP}$ spreads the execution evenly over the time horizon $t\in [0,T]$. Specifically, the TWAP liquidation strategy is given by  

\begin{equation}
    X_t^{\text{TWAP}} = \frac{X_0(T-t)}{T}, \ \  t\leq T,
\end{equation}\noindent
and at maturity $X_T^{TWAP} = 0$. So that the agent's trading rate is constant (constant velocity) and $\xi_t^{TWAP} \equiv \frac{X_0}{T}$ for all $t \in [0,T]$. It has been shown in the seminal work of Almgren and Chriss \cite{almgren2000optimal} that under the martingale assumption on the execution price dynamics, linear permanent impact,  and excluding risk-aversion such a strategy is optimal in terms of minimizing the expected quadratic execution costs.

Since we are considering the framework where apart from the permanent and the temporary market impact component, we also have a transient component in the impact of the trading rate, we will also consider as a strategy the one derived by Obizhaeva and Wang in \cite{Obizhaeva2005}.

\begin{equation}
    dX_t^{OW} = - \frac{X_0}{2+ \rho T}\left[\delta_0(dt) + \rho + \delta_T(dt) \right],
\end{equation}
where $\delta$ is simply the Dirac function. Or similarly one can see the position of the liquidation strategy as 

\begin{equation}
X_t \rightarrow x_0 - \frac{x_0}{2+\rho T}\left[ \delta_0(t) + \rho t + \delta_T(t)\right].
\end{equation}
So essentially the strategy starts and ends with two block trades and then liquidates the position on a constant rate as in the case of the TWAP strategy. The only difference is on the resilience speed $\rho$ of the transient part of the price. In the particular case where $\rho  \to \infty$ meaning essentially that the transient market impact part of the strategy diminishes then the optimal strategy boils down to a TWAP. 

One of the most commonly used benchmark strategies in algorithmic trading is the Volume-Weighted Average Price (VWAP). It is straightforward to compute \emph{ex post}, making it a popular performance metric: if execution prices remain close to the VWAP, then the trader has effectively matched the volume-weighted average price over the execution window. However, implementing a VWAP strategy in real time is challenging because the intraday distribution of total traded volume is unknown in advance. Although volume in intraday markets often increases toward gate closure \cite{KZ20}, substantial variability persists both intraday and across different days.

Formally, the VWAP between times $T_1$ and $T_2$ is given by
\[
    \mathrm{VWAP}(T_1,T_2)
    = \frac{\int_{T_1}^{T_2} S_t \, dV_t}{\int_{T_1}^{T_2} dV_t},
\]
where $S_t$ is the mid-price at time $t$ and $V_t$ the total traded volume up to time $t$. The VWAP emphasizes trades that occur at times of higher volume, arguably capturing a more ``representative'' daily price than the Time-Weighted Average Price (TWAP), which does not weight trades by volume.

Targeting a VWAP execution schedule is difficult due to the random nature of traded volume throughout the day. One practical approach is to submit orders at a rate proportional to the ongoing market order flow, approximating a fraction of the total trading activity \cite{CJ16}. Because we model trading intensity, a VWAP-like strategy can be formulated by partitioning the trading horizon into minute buckets and distributing trades in proportion to the historical (or forecast) fraction of volume within each bucket.

Concretely, let \( \hat{v}_k \) represent the empirically observed volume in bucket \( k \), and let \( N \) denote the total number of buckets. We define re-scaling factors
\[
F_k
:= \frac{\hat{v}_k}{\frac{1}{N}\,\sum_{\ell=1}^N \hat{v}_\ell},
\quad
\text{ensuring}
\;\;\prod_{k=1}^N F_k = 1.
\]
Hence, the cumulative quantity \(X_0\) to be liquidated is partitioned into shares
\[
   \xi_k
   := \left( \frac{X_0}{N} \right) F_k,
\quad
X_k
= (N - k)
\left(
   \frac{X_0}{N} F_k
\right),
\]
where \( X_k \) represents the remaining inventory after $k$ buckets. This approach mimics the proportion of volume traded in each bucket, thus approximating a real-time VWAP schedule. In principle, more advanced forecasting methods for intraday volume could improve performance by adapting to evolving market conditions.

Figure~\ref{fig:ResidualPositionsCombined} shows an example of the resulting trading trajectory compared with TWAP and the Hawkes-derived optimal strategy, assuming $X_0 = 250\,\mathrm{MWh}$ for liquidation across four representative hourly contracts (15h-21h). The paths illustrate how the VWAP strategy seeks to align trading with the estimated intraday volume distribution, while TWAP divides the trading horizon evenly, and the Hawkes-based method exploits predicted intensity dynamics for potentially lower costs.

Recall that the actual traded volume $v_k$ is stochastic and not known \emph{a priori}, so we use the empirical values $\hat{v}_k$ derived from the LOB data. Much like equity markets, where trading activity often follows a U-shape with peaks near the open and close auctions, intraday energy markets exhibit a diurnal pattern. Consequently, we estimate the historical traded volume for each product in one-minute intervals (yielding 480 intervals over the last 8 hours of trading). 

To ensure a sufficient proportion of trading volume in each bucket, we require that 
$\displaystyle \prod_{k=1}^N F_k = 1$. Of course, instead of using a fixed historical mean, more advanced trading-volume forecasting methods could improve performance by incorporating an opportunistic component into the strategy.

\begin{figure*}[!htbp]
    \centering
    \begin{subfigure}[b]{0.45\textwidth}
        \centering
        \includegraphics[width=\textwidth]{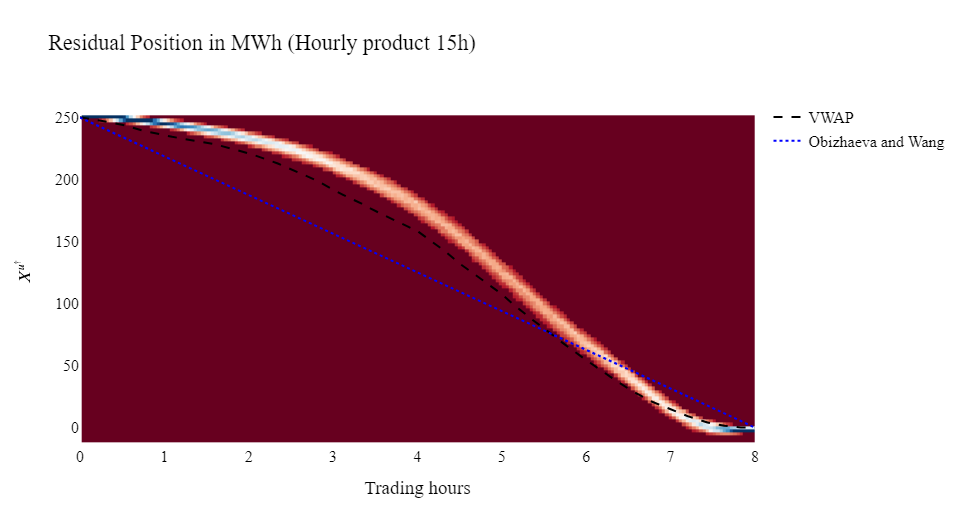}
        \caption{Residual position for the 15h product.}
        \label{fig:residualpos1}
    \end{subfigure}
    \hfill
    \begin{subfigure}[b]{0.45\textwidth}
        \centering
        \includegraphics[width=\textwidth]{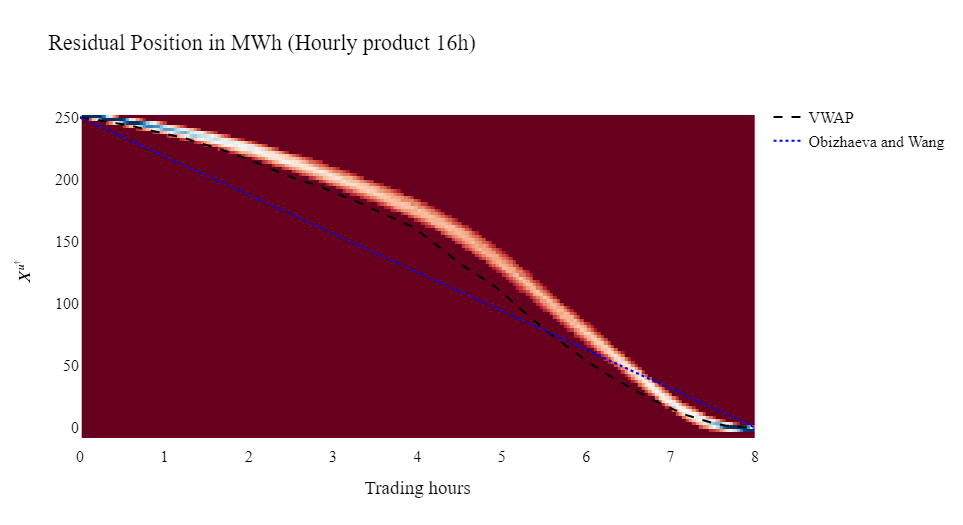}
        \caption{Residual position for the 16h product.}
        \label{fig:residualpos2}
    \end{subfigure}
    
    \vspace{1em} 
    
    \begin{subfigure}[b]{0.45\textwidth}
        \centering
        \includegraphics[width=\textwidth]{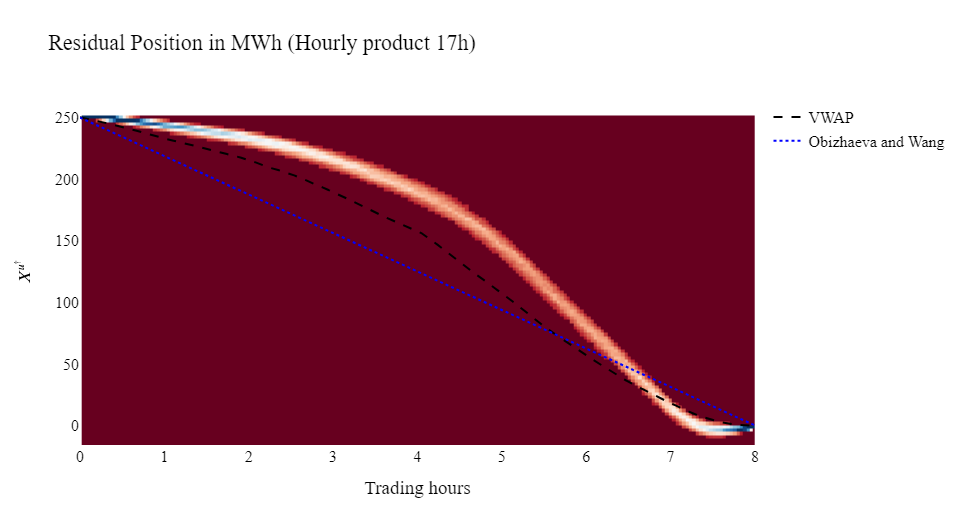}
        \caption{Residual position for the 17h product.}
        \label{fig:residualpos3}
    \end{subfigure}
    \hfill
    \begin{subfigure}[b]{0.45\textwidth}
        \centering
        \includegraphics[width=\textwidth]{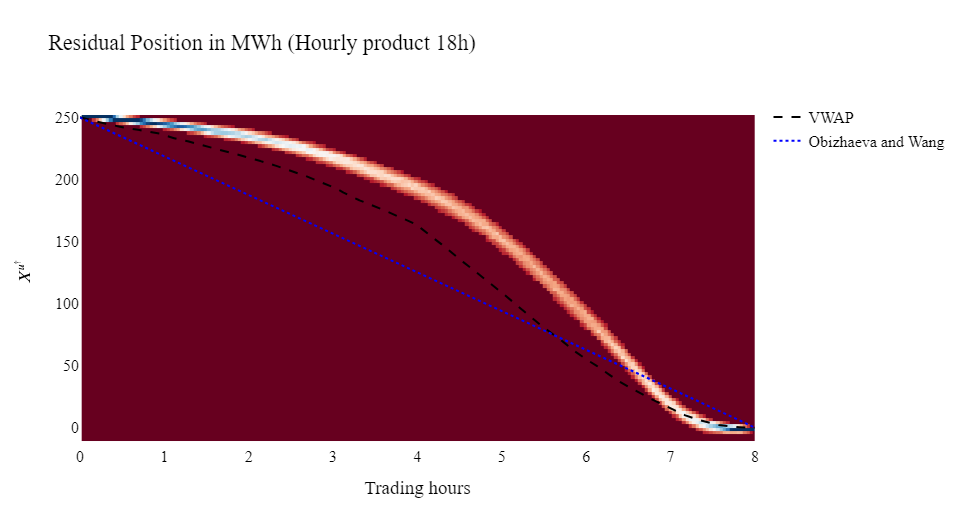}
        \caption{Residual position for the 18h product.}
        \label{fig:residualpos4}
    \end{subfigure}
        \vfill
    \begin{subfigure}[b]{0.45\textwidth}
        \centering
        \includegraphics[width=\textwidth]{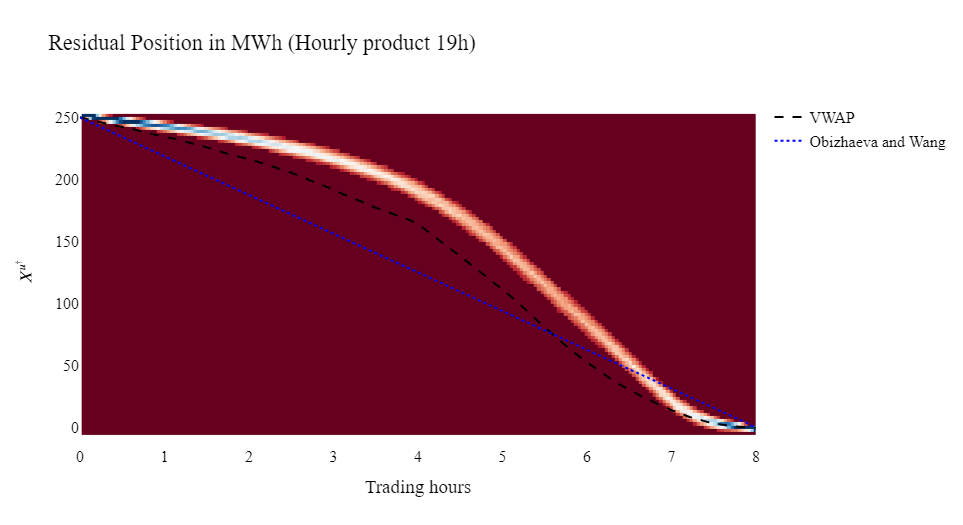}
        \caption{residual position for the 19h product.}
        \label{fig:plot19}
    \end{subfigure}
    \hfill
    \begin{subfigure}[b]{0.45\textwidth}
        \centering
        \includegraphics[width=\textwidth]{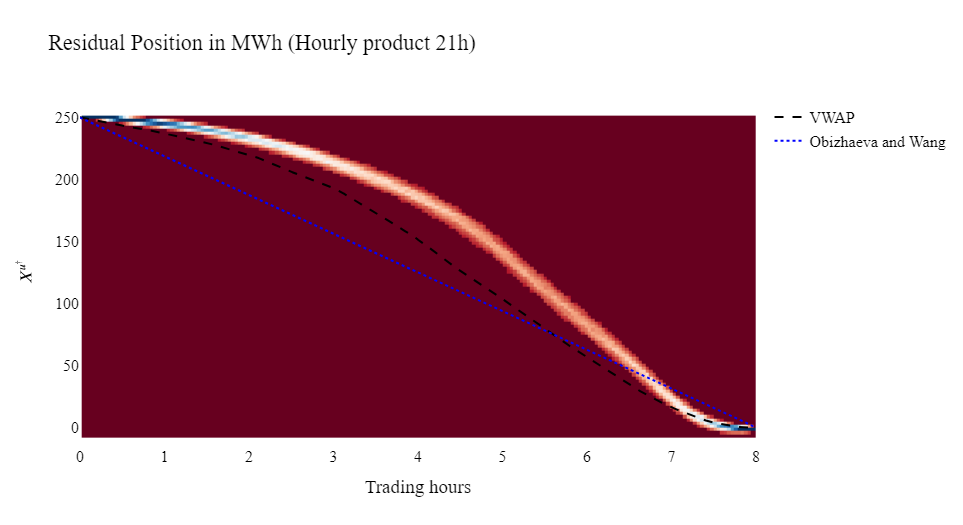}
        \caption{Residual position for the 21h product.}
        \label{fig:plot20}
    \end{subfigure}
    \caption{Trading trajectories (Residual position) for the optimal execution strategy (simulations), TWAP, and VWAP for the last 8 trading hours assuming a total of $X_0 = 250$ MWh to be liquidated. The simulations are based on parameters fitted for the 15h, 16h, 17h, 18h, 19h, 21h products. }
    \label{fig:ResidualPositionsCombined}
\end{figure*}

One can easily understand that different trading algorithms result in different trading trajectories. Having presented the mathematical formulation of our optimal execution strategy as well as the benchmark strategies that we consider for the backtesting of the optimal strategy, we will present some numerical results. The most straightforward way to understand volume allocation behaviour and the metaorder execution problem (in terms of a liquidation problem) is by inspecting a trade path plot. Figure~\ref{fig:ResidualPositionsCombined} showcases a MC simulation of the optimal trading trajectory, the VWAP and TWAP strategy. Considering that positions 200-300 MWh apply to many market participants and depict an industry-wide optimal execution problem, we consider that the trader wants to liquidate a total of 250 MWh in the ID market in a time window of T = 8h. We simulate the optimal strategy using a total of $N_{sims} =  1000$.

For this simulation we consider the calibrated parameters for the 15h, 16h, 17h, 18h, 19h and 21h products as well as the empirical volumes calculated for this hourly products. As we can see from the figure the VWAP strategy starts trading with slower trading rates and leaves a higher position open in the early/mid-trading phases. Just before the switch from pan-European XBID to local trading, we can see that positions are closed more aggressively leaving a smaller position in the last trading hour. Finally, the optimal execution strategy follows the trading trajectory of the VWAP strategy, however being even more passive at the start of the trading phase and being aggressive during the end. For the last trading hour, we can see how the strategy resembles even more than one of a VWAP execution type. Taking a look at the trading trajectories of the optimal strategy $X_t^*$ and the $X_t^{VWAP}$ these trading paths avoid both the higher BAS in the early phase of the trading (far from delivery) but also the high BAS levels in local trading. As we have already described Figure \ref{fig:TempImpactCost} showcases the mean instantaneous market impact coefficient with respect to the time to gate closure, showcasing the potential profitability of the optimal execution strategy, when our objective is to minimize the transaction costs of the liquidation strategy. 

\subsection{Backtesting and Transaction Cost Analysis }

In this section we will describe the back-testing methodology and show the performance of the optimal trading strategy compared to the benchmark VWAP and TWAP strategies. The backtesting window we considered ranges from 1st of January 2023 up to 1st of March 2024. Let's assume for now that we have a generic liquidating strategy $\xi_t^{\text{Strat}}$ and our goal is to liquidate a total of $X_0$ MWh over a fixed specified time horizon $T$. If not specified differently we set $T=8h$, so that the trader wants to liquidate his position in a fixed period of 8 hours before gate closure. We assume that the trader chooses to execute the metaorder via a sequence of child orders with size $\xi_t^{\text{Strat}}$ depending on the liquidating strategy he considers and assume that $\sum_{t=1}^T \xi_t^{\text{Strat}} = X_0$. The backtesting of the optimal execution problem first requires setting a suitable loss-function to minimize. One possible example is the average execution shortfall (or slippage), $
    \Delta ^{\text{Strat}}= \sum_{t=1}^T \xi_t^{\text{Strat}}(p_t - m_0),
$
which measures the total price paid for executing the metaorder, relative to the initial mid-price $m_0$, including the bid-ask spread. Under the assumption and the modeling framework we have chosen for the specification of the mid-price dynamics (linear propagator/transient impact kernel), we can actually write the average execution shortfall $ \mathcal{ D}^{\text{Strat}}:=\mathbb{E}\left[\Delta^{(Strategy)}\right]$ as

\begin{align}
\label{shortfall}
  \mathcal{ D}^{\text{Strat}}&= \!\sum_{t=1}^T \xi_t^{\text{Strat}}\big( \sum_{1\leq t'<t} G(t-t')\xi_{t'}^{\text{Strat}}\big)\!+\! \sum_{t=1}^T\frac{\xi_t^{\text{Strat}} s_t}{2}
\end{align}
where we have used that $p_t = m_t + s_t/2$, with $s_t/2$ the half-spread. We have to note that in our modeling framework the sign of the metaorder is uncorrelated with the order flow from the rest of the market.  We refer the reader to the following work \cite{CHTH23} for an optimal trade execution framework with endogenous order flow, where a trader’s own
order submissions trigger child orders and further influence future price dynamics. Equation \eqref{shortfall} shows that the expected execution shortfall is partly due to the spread and partly due to the impacts of past trades. Our goal is to use the expected shortfall as our cost function to back-test the optimal execution strategy to realized market data, by minimizing this cost functional. We will also consider the simplified case where the transient market impact component boils down to the instantaneous impact only which we have estimated using the procedure described in Section \ref{ExecutionCostEstimates}. For the half-spread cost we use the empirical time-weighted half-spread directly obtained from the LOB data as described in Section~\ref{ExecutionCostEstimates}.

We will implement the optimal execution strategy for each hourly product using the estimated parameters we have derived, both for the Hawkes process which models the arrival rates of trades (buy/sell). We are going to distinguish between the univariate and bivariate calibration using the obtained parameters we presented in Table~\ref{tab:combinedParams} for the univariate case and Table~\ref{tab:combined-table} for the bivariate case accordingly. The transaction cost analysis results, will be devoted to each hourly product separately. Our goal is to test whether the optimal strategy allows for cost reduction when we are considering the liquidation of a fixed position of total $X_0 = 200-300$ MWh. The relative performance measures with respect to the optimal strategy is given by $r_{\text{Bench}} =  ( \mathcal{  D}^{\text{Bench}} -  \mathcal{  D}^{\text{Strat}})/  \mathcal{  D}^{\text{Bench}}$. Starting from the univariate case, where the arrival of buy and sell market orders are independent, Figure \ref{fig:RelativeCostImprov} showcases the relative cost improvements (in percentage points) for each hourly trading product in comparison to the TWAP and VWAP strategies. The relative improvements in terms of cumulative execution costs are on average positive, indicating a clear cost reduction when the trader chooses to liquidate their position using the optimal execution strategy.
\begin{figure*}[!htbp]
    \centering
    \includegraphics[width=0.9\linewidth]{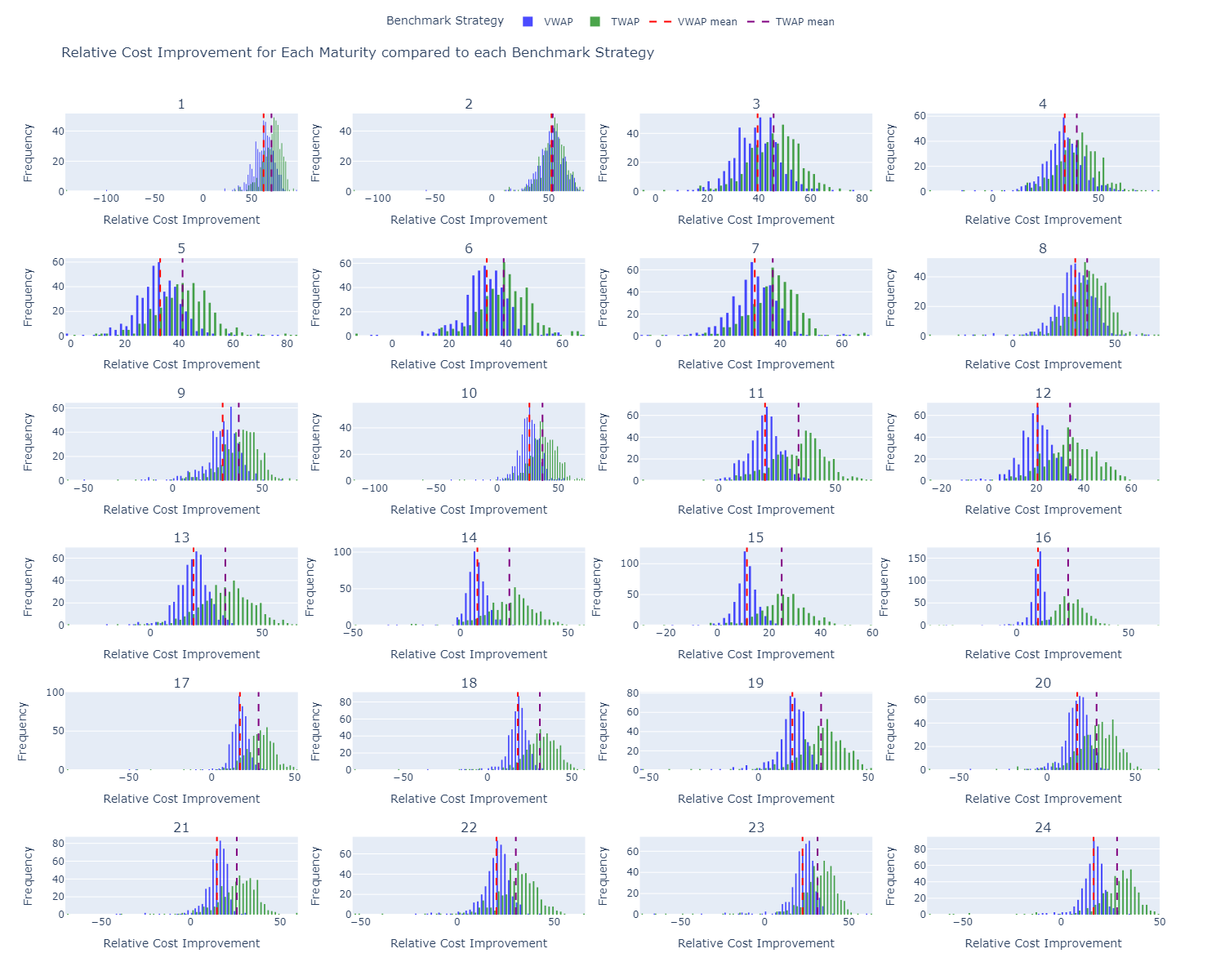}
    \caption{Relative Cost Improvement in terms of cumulative costs (\%) for each hourly product $r_{BenchMark}$. Vertical lines showcase the mean relative improvement of the optimal trading strategy compared to each benchmark strategy (VWAP/TWAP). }
    \label{fig:RelativeCostImprov}
\end{figure*}
For every trading hour, the relative improvement compared to the TWAP strategy is always greater than the improvement compared to the VWAP liquidation strategy. This likely occurs because the optimal trading trajectories derived from the VWAP strategy are closer to those obtained from the optimal strategy (see Figure \ref{fig:ResidualPositionsCombined}), whereas the TWAP strategy liquidates the position at a constant rate.

Examining the average relative cost improvements associated with the optimal strategy and the two benchmarks for each maturity, we observe from Figure \ref{fig:relative-cost-improvement-plot} and Table \ref{tab:relative-cost-improvement} that the optimal strategy achieves cost reduction for each product. Specifically, the cost reduction is very significant for the first hourly trading products (1h-4h) and then stabilizes for the rest of the trading hours, exhibiting a slight decrease for the mid-day trading products (12h-16h).

Considering the seasonality of liquidity across each trading product, as shown in Figure~\ref{fig:VarIntensity1}, there is a clear pattern in the associated liquidity of each hourly product. The most liquid products are associated with mid-day (noon) hours, while the first hourly trading products are the least liquid ones. This suggests that the amount of cost savings is negatively correlated with the average traded volume observed for each product. Therefore, comparatively less volatile, more liquid trading products at lower tick sizes may benefit less from the optimal trading strategy compared to the less liquid products, where the improvement is more pronounced.

Considering the \emph{time-dependent} (pathwise, not only terminal) behavior of trading costs, 
Figure~\ref{fig: RelativeImprTimeDepe} displays the surface of the relative cost difference across the trading hours (last 8 trading hours in our backtests) and hourly trading product ($1,\dots,24$). By construction, positive values indicate that the benchmark incurs a higher cumulative cost than the optimal strategy at that point in time. For completeness, terminal (end-of-horizon) relative improvements are reported in 
Figures~\ref{fig:relative-cost-improvement-plot}. In both benchmarks, the optimal strategy tends to incur smaller costs early in the session—consistent with its slower initial execution speed—and then accelerates as time to maturity shrinks. 
Immediately after the market decoupling into the \emph{local-trading} regime, we observe a temporary dip in the mean relative improvement, reflecting wider bid–ask spreads and thinner liquidity; 
if an insufficient fraction of the position is completed before decoupling, the instantaneous cost increases. 
Nevertheless, the strategy recovers and the terminal relative improvements at maturity ($T=8$\,h) remain positive against both TWAP and VWAP.

\begin{table}[!htbp]
    \centering
    \scriptsize
    \setlength{\tabcolsep}{3pt}
    \renewcommand{\arraystretch}{1.1}
    \begin{tabular}{>{\centering\arraybackslash}p{0.6cm}|
                    >{\centering\arraybackslash}p{2.7cm}|
                    >{\centering\arraybackslash}p{2.7cm}}
        \toprule
        \textbf{Hour} & \textbf{VWAP Mean ± Std (\%)} & \textbf{TWAP Mean ± Std (\%)} \\
        \midrule
        1  & 62.64 ± 12.57 & 70.73 ± 13.65 \\ 
        2  & 52.48 ± 11.69 & 53.55 ± 13.14 \\ 
        3  & 39.57 ± 8.69  & 45.77 ± 10.14 \\
        4  & 34.09 ± 9.48  & 39.80 ± 10.94 \\
        5  & 33.07 ± 8.93  & 41.35 ± 10.52 \\
        6  & 33.32 ± 8.10  & 39.28 ± 9.37  \\
        7  & 31.56 ± 8.07  & 37.43 ± 8.42  \\
        8  & 30.62 ± 9.97  & 36.47 ± 11.29 \\
        9  & 27.95 ± 10.58 & 36.97 ± 12.36 \\
        10 & 26.46 ± 9.18  & 37.16 ± 14.21 \\
        11 & 19.84 ± 6.75  & 34.24 ± 11.68 \\
        12 & 20.52 ± 7.26  & 34.24 ± 10.97 \\
        13 & 19.33 ± 7.05  & 33.56 ± 12.33 \\
        14 & 7.96 ± 5.33   & 22.81 ± 11.52 \\
        15 & 11.44 ± 4.48  & 24.90 ± 9.50  \\
        16 & 9.60 ± 3.88   & 23.06 ± 8.79  \\
        17 & 17.13 ± 6.33  & 28.33 ± 10.98 \\
        18 & 18.06 ± 6.41  & 31.14 ± 11.67 \\
        19 & 15.52 ± 7.20  & 28.61 ± 11.33 \\
        20 & 17.23 ± 7.35  & 28.19 ± 12.38 \\
        21 & 14.75 ± 7.18  & 25.84 ± 11.89 \\
        22 & 20.23 ± 7.27  & 30.23 ± 12.12 \\
        23 & 22.54 ± 9.34  & 31.47 ± 13.75 \\
        24 & 16.47 ± 5.77  & 28.38 ± 12.58 \\
        \bottomrule
    \end{tabular}
    \caption{Relative cost improvement statistics (\%) (mean ± std) compared to each benchmark strategy for each maturity in the model where there is only self-excitation between buy and sell MOs.}
    \label{tab:relative-cost-improvement}
\end{table}


\begin{figure*}[t]
    \centering
    \begin{subfigure}[b]{0.45\textwidth}
    \centering
    \includegraphics[width=\textwidth]{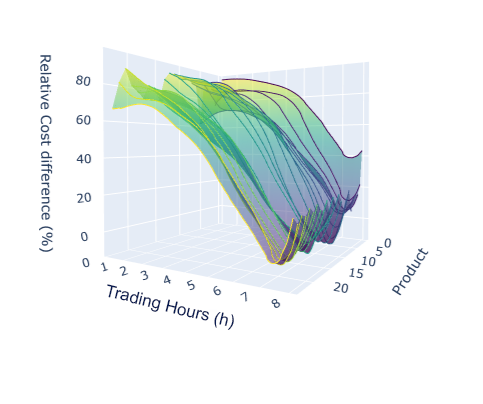}
    \label{fig:TimeDependentCostImprovementTWAP}
    \end{subfigure}
\hfill
    \begin{subfigure}[b]{0.45\textwidth}
    \centering
    \includegraphics[width=\textwidth]{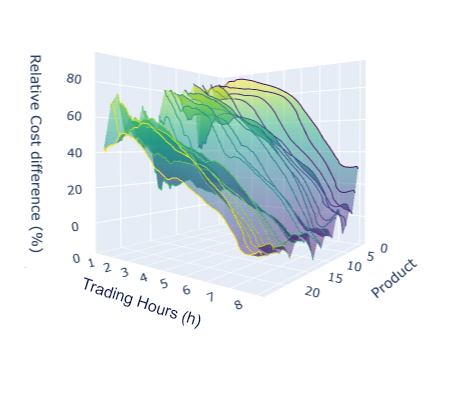}
    \label{fig:TimeDependentCostImprovementVWAP}
    \end{subfigure}
    \caption{Relative Cost Difference between the optimal trading strategy and the TWAP (left)/VWAP (right)  for different hourly Trading products and with respect to time to maturity. A positive value means the TWAP/VWAP has a higher cost, and vice versa.}
    \label{fig: RelativeImprTimeDepe}
\end{figure*}\noindent

Extending the model to the multivariate case, our goal is to test whether considering the full impact kernel with the cross-excitation coefficients, indeed achieves lower execution costs. Since there is not a observable difference in terms of the trading trajectories we just present the results in terms of the transaction costs of the optimal strategy. Figure~\ref{fig:RelativeCostImprov2} shown in the Appendix \ref{sec:complementaryRes} showcases the relative cost improvements (in percentage) for each hourly trading product, similar to the univariate case. The relative improvements in terms of cumulative execution costs are again always positive, however, it is not clear whether the improvement compared to the univariate case is significant. 
\begin{figure}[H]
    \centering
    \includegraphics[width=\linewidth]{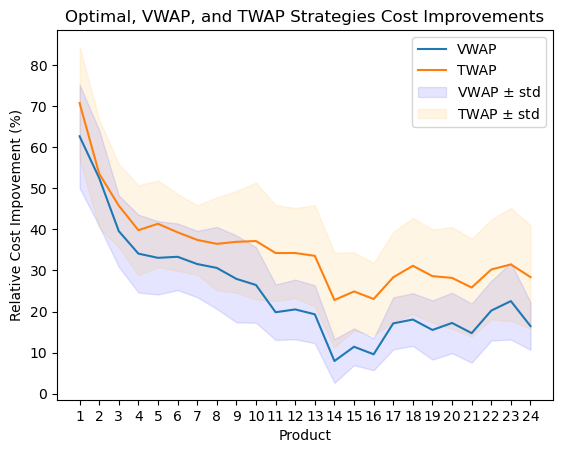}
    \caption{Comparison of relative cost improvements across maturities.}
    \label{fig:relative-cost-improvement-plot}
\end{figure}
Looking at the mean relative cost improvements of the bivariate model compared to the VWAP, TWAP and the univariate model Table \ref{tab:relative-cost-improvement2} and Figure \ref{fig:relative-cost-improvement-plot2} showcase that the results compared to the benchmark strategies are again significant. However, compared to the univariate case, where we have only the self-excitation parameters in the exponential kernel, the results vary across each hourly trading product, indicating that there is no clear improvement when considering the full kernel for the modeling of the order flow of buy and sell MOs. 

\begin{table}[!htbp]
    \centering
    \scriptsize
    \setlength{\tabcolsep}{2pt}
    \renewcommand{\arraystretch}{1.1}
    \begin{tabular}{>{\centering\arraybackslash}p{0.6cm}|
                    >{\centering\arraybackslash}p{2.3cm}|
                    >{\centering\arraybackslash}p{2.3cm}|
                    >{\centering\arraybackslash}p{2.3cm}}
        \toprule
        \textbf{Hour} & \textbf{VWAP Mean ± Std} & \textbf{TWAP Mean ± Std} & \textbf{Univariate Mean ± Std} \\
        \midrule
        1  & 54.74 ± 11.85 & 64.67 ± 12.83 & -23.17 ± 7.46 \\
        2  & 50.97 ± 12.82 & 52.11 ± 14.13 & -2.82 ± 4.47 \\
        3  & 38.70 ± 9.05  & 44.98 ± 10.52 & -1.38 ± 2.42 \\
        4  & 34.35 ± 10.34 & 40.00 ± 11.77 & \cellcolor{red!25}0.61 ± 1.85 \\
        5  & 32.08 ± 8.73  & 40.47 ± 10.63 & -1.60 ± 3.10 \\
        6  & 29.89 ± 8.51  & 36.12 ± 10.01 & -5.18 ± 2.67 \\
        7  & 35.95 ± 8.91  & 41.42 ± 9.14  & \cellcolor{red!25}6.67 ± 2.79 \\
        8  & 32.79 ± 10.10 & 38.43 ± 11.47 & \cellcolor{red!25}3.21 ± 1.83 \\
        9  & 28.07 ± 9.38  & 37.03 ± 11.78 & -0.07 ± 4.28 \\
        10 & 25.44 ± 8.72  & 36.26 ± 14.29 & -1.48 ± 2.41 \\
        11 & 21.13 ± 7.04  & 35.27 ± 11.81 & \cellcolor{red!25}1.64 ± 1.27 \\
        12 & 18.69 ± 7.24  & 32.57 ± 12.09 & -2.42 ± 4.56 \\
        13 & 15.30 ± 6.09  & 30.20 ± 12.73 & -5.15 ± 2.80 \\
        14 & 13.69 ± 5.04  & 27.51 ± 11.61 & \cellcolor{red!25}6.20 ± 2.06 \\
        15 & 13.46 ± 4.78  & 26.57 ± 9.86  & \cellcolor{red!25}2.28 ± 2.36 \\
        16 & 12.00 ± 4.55  & 24.93 ± 10.13 & \cellcolor{red!25}2.64 ± 3.26 \\
        17 & 15.39 ± 6.11  & 26.81 ± 10.96 & -2.14 ± 2.46 \\
        18 & 15.46 ± 4.98  & 28.98 ± 11.17 & -3.31 ± 2.26 \\
        19 & 13.62 ± 5.80  & 26.79 ± 12.33 & -2.54 ± 5.59 \\
        20 & 15.64 ± 6.05  & 26.78 ± 12.32 & -2.12 ± 3.79 \\
        21 & 13.52 ± 5.40  & 24.66 ± 11.99 & -1.67 ± 4.15 \\
        22 & 15.56 ± 6.27  & 26.17 ± 11.99 & -6.04 ± 3.20 \\
        23 & 18.45 ± 7.72  & 27.88 ± 12.60 & -5.57 ± 3.94 \\
        24 & 17.37 ± 5.80  & 29.14 ± 12.63 & \cellcolor{red!25}1.09 ± 0.61 \\
        \bottomrule
    \end{tabular}
    \caption{Relative cost improvement statistics (\%) (mean ± std) compared to each benchmark strategy for each maturity in the model where both self- and cross-excitation are considered.}
    \label{tab:relative-cost-improvement2}
\end{table}

\begin{figure}[!htbp]
    \centering
    \includegraphics[width=\linewidth]{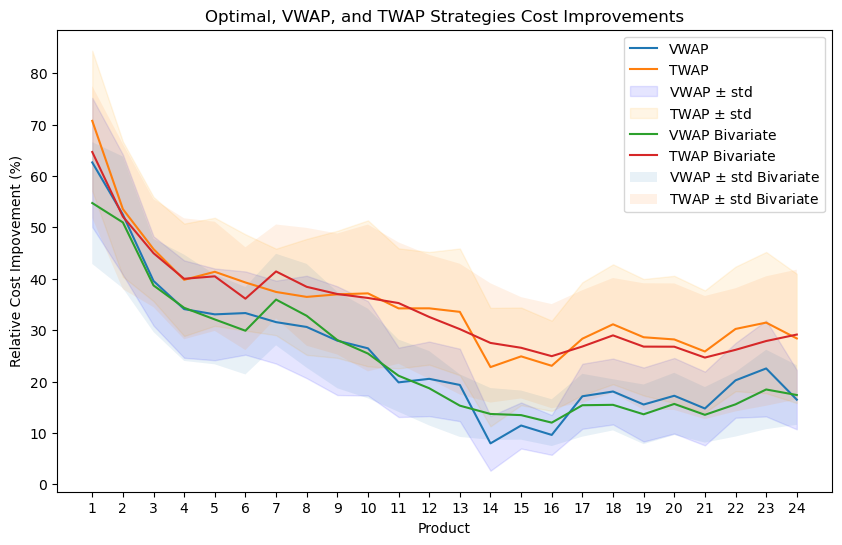}
    \caption{Comparison of relative cost improvements across maturities under self- and cross-excitation.}
    \label{fig:relative-cost-improvement-plot2}
\end{figure}

\section{Conclusion}
\label{conclusion}

In this study we investigate optimal execution strategies in the German Intraday Electricity Market by employing a Hawkes process for the modeling of the flow of market orders and a linear transient impact that decays exponentially to account for the resilience of the market. We introduced a novel calibration protocol, leveraging spline-based baseline intensity functions for the Hawkes process, effectively capturing intraday seasonality and trading activity dynamics throughout a trading session. 

A calibration protocol for the model is introduced using tick-by-tick LOB data for the German ID market. We also detail the procedures of estimating each types of market impact and transaction costs using raw LOB data. A series of backtests on the hourly products shows that the optimal strategy achieves significant cost reductions over traditional execution strategies, like TWAP and VWAP. Analysis on each individual product revealed that cost reductions are particularly substantial for early trading hourly products, and stabilized thereafter, with a slight decrease in the mid-day products where on average trading activity is higher. This indicates that cost savings are negatively correlated with
the average traded volume for each product, suggesting that less volatile, more liquid products
might benefit less from the optimal strategy compared to less liquid ones.

There are several potential extensions to the methodology and modeling framework presented
in this work. One possible extension is to model multiple contracts simultaneously rather than focusing on a single contract. Our current one dimensional optimal execution model does not capture the dependence between different maturities, which is crucial for valuing assets with payoffs depending on several maturities (e.g. battery valuation). This could be addressed by employing a multivariate transient price impact kernel which models how trades for a given trading hourly product would (directly or indirectly) impact the price of another product. 

One could also formulate an optimal execution problem under a cross-impact propagator model, where in intraday markets overlapping sessions exist for different delivery days. Given the strong interdependencies among products, incorporating cross-product and cross-impact effects may offer deeper insights into price formation in intraday electricity markets.

\section*{Acknowledgement}
We sincerely thank the two anonymous referees for their constructive feedback, which greatly improved the clarity and presentation of this paper.
Konstantinos Chatziandreou gratefully acknowledges the support through Statkraft Trading GmbH. Konstantinos Chatziandreou is grateful to Daniel Gruhlke and Simon Hirsch for many helpful discussions. This work contains the author’s opinion and does not necessarily reflect Statkraft’s position. The authors declare no conflict of interest.

\section*{Data Availability Statement}
The data that support the findings of this study are available from EPEX SPOT. The data are available from
the corresponding author upon reasonable request and if permission from EPEX SPOT is given.

\bibliographystyle{plainnat}
\bibliography{Bib.bib}

\appendix{}

\section{MLE for Hawkes Parameters}
\label{sec:MLEHawkes}

We now consider a $P$-variate multivariate Hawkes process in a deseasonalized setting. Let
\[
\mathbf{N}_t 
= \bigl(N_t^{(1)}, \dots, N_t^{(P)}\bigr)
\]
be a $P$-dimensional point process, whose $p$-th component has the conditional intensity
\begin{equation}\label{eq:HawkesProcessConditionalIntensity}
    \lambda_t^{(p)}
    = \lambda_{\infty}^{(p)}
    + \sum_{m=1}^P \int_{-\infty}^t \phi_{p,m}(t - u)\, dN_u^{(m)},
    \quad p = 1, \dots, P,
\end{equation}
where $\bm{\lambda}_\infty 
= \bigl(\lambda_\infty^{(1)}, \dots, \lambda_\infty^{(P)}\bigr)$ 
is the baseline intensity vector, and $\bm{\phi}(u) 
= \bigl(\phi_{p,m}(u)\bigr)_{p,m}$ is a nonnegative kernel matrix such that $\phi_{p,m}(u) = 0$ for all $u < 0$. The intensities $\lambda_t^{(p)}$ may exhibit both self- and cross-excitation, depending on whether $\phi_{p,m}$ is zero or nonzero for various indices $p \neq m$.

We estimate the Hawkes parameters via maximum likelihood (MLE). Suppose we observe the process over $[0, T]$, and let
\begin{align*}
\mathcal{T}_{(p)} 
= \bigl\{ t_1^{(p)}, \dots, t_{N_T^{(p)}}^{(p)} \bigr\}, 
\quad p = 1, \dots, P,    
\end{align*}
denote the ordered event times for each component $p$, where $N_T^{(p)}$ is the total number of events in component $p$ by time $T$. Denote by
\begin{align*}
\bm{\mathcal{T}}
= \big(\,\mathcal{T}_{(1)}, \dots, \mathcal{T}_{(P)}\big)    
\end{align*}
the full collection of events across all $P$ components.

The log-likelihood of a multivariate Hawkes process is then
\begin{equation}\label{GeneralLikelihood}
    \log \mathcal{L}\bigl(\Theta;\,\bm{\mathcal{T}}\bigr)
    = \sum_{p=1}^P \Big[
    \sum_{j=1}^{N_T^{(p)}} 
    \log \lambda_{\,t_j^{(p)}}^{(p)}
    \int_0^T \lambda_u^{(p)} \, du
    \Big],
\end{equation}
where $\Theta$ represents the parameter set (e.g., baseline intensities, kernel parameters). The first term sums the logarithms of the intensities evaluated at each observed event time, while the second term integrates the intensity over $[0,T]$. 

\begin{proposition}\label{MLE-Multivariate Hawkes}
We define $\Delta_{T,k}^n := T - t_k^n$ and $\Delta_{k,i}^{mn} := t_k^m - t_i^n$, and let the kernel function of the Hawkes process be of exponential form. Then the conditional intensity function in Equation~\eqref{eq:HawkesProcessConditionalIntensity} reads
\begin{equation}
 \lambda^{(p)}_t = \lambda^{(p)}_{\infty} + \sum_{m=1}^P \sum_{j=1}^{N^m(t)} \alpha_{pm}\exp\left(-\beta_{pm}(t-t_j^m)\right),
\end{equation}
where $\bm{\alpha} = \left(\alpha_{pm}\right)$ and $\bm{\beta} = \left(\beta_{pm}\right)$ are the excitation and decay parameters. In this particular case, the Markovianity of the stochastic process $\left(\bm{N},\bm{\lambda}\right)$ is guaranteed. This allows one to compute the likelihood estimator, gradient, and Hessian in a recursive way. In particular, we have:
\begin{align*}
    \log \mathcal{L}^{m} \left(\Theta;\bm{\mathcal{T}}\right) 
    &= - \lambda_{\infty}^m T 
    \!- \sum_{n=1}^P \frac{\alpha_{mn}}{\beta_{mn}}\!\sum_{t_k^n < T} (1 \!-\! \E^{-\beta_{mn} \Delta_{T,k}^n}) \\
    &\quad\! + \sum_{t_k^m < T} \log\big(\lambda_{\infty}^m \!+ \sum_{n=1}^P \alpha_{mn} R_{mn}(k)\big),
\end{align*}
where $R_{mn}(k) \!= \sum_{i: t_i^n < t_k^m} \exp\left(-\beta_{mn} \Delta_{k,i}^{mn} \right)$ for $k \geq 2$, and $R_{mn}(1) = 0$. This can be defined recursively as
\begin{align*}
    R_{mn}(k) &= \exp\left(-\beta_{mn} \Delta_{k,k-1}^{mm}\right) R_{mn}(k-1) 
    \\ 
    &\quad + \sum_{i: t_{k-1}^m < t_i^n < t_k^m} \exp\left(-\beta_{mn} \Delta_{k,i}^{mn} \right).
\end{align*}

The gradient of the log-likelihood with respect to the parameters is given by:
\begin{align*}
    \frac{\partial L^m}{\partial \lambda^m_{\infty}} &= -T + \sum_{k: t_k^m < T} \frac{1}{\lambda^m_{\infty} + \sum_{j=1}^P \alpha_{mj} R_{mj}(k)}, \\
    \frac{\partial L^m}{\partial \alpha_{mn}} &= - \frac{1}{\beta_{mn}} \sum_{k: t_k^n < T} \left[1 - \exp\left(-\beta_{mn} \Delta_{T,k}^n\right)\right]\\
    &\quad + \sum_{k: t_k^m < T} \frac{R_{mn}(k)}{\lambda^m_{\infty} + \sum_{j=1}^P \alpha_{mj} R_{mj}(k)}, \\
    \frac{\partial L^m}{\partial \beta_{mn}} &= 
    \frac{\alpha_{mn}}{\beta_{mn}^2} \sum_{k: t_k^n < T} \left[1 - \exp\left(-\beta_{mn} \Delta_{T,k}^n\right)\right] \\
    &\quad - \frac{\alpha_{mn}}{\beta_{mn}} \sum_{k: t_k^n < T} \Delta_{T,k}^n \exp\left(-\beta_{mn} \Delta_{T,k}^n\right)\\ 
    &\quad - \sum_{k: t_k^m < T} \frac{\alpha_{mn} R_{mn}'(k)}{\lambda^m_{\infty} + \sum_{j=1}^P \alpha_{mj} R_{mj}(k)},
\end{align*}
where $R_{mn}'(1) = 0$ and
\[
R_{mn}'(k) := \sum_{i: t_i^n < t_k^m} \Delta_{k,i}^{mn} \exp\left(-\beta_{mn} \Delta_{k,i}^{mn} \right).
\]

The Hessian terms include:
\begin{align*}
    \frac{\partial^2 \mathcal{L}^m}{\partial (\lambda_{\infty}^m)^2} 
    &= - \sum_{k: t_k^m < T} \left[ \frac{1}{\lambda^m_{\infty} + \sum_{j=1}^P \alpha_{mj} R_{mj}(k)} \right]^2,
    \end{align*}

    \begin{align*}
    \frac{\partial^2 \mathcal{L}^m}{\partial \lambda_{\infty}^m \partial \lambda_{\infty}^n} &= 0, \quad m \neq n,
    \end{align*}
    
     \begin{align*}
    \frac{\partial^2 \mathcal{L}^m}{\partial \alpha_{mn}^2} 
    &= - \sum_{k: t_k^m < T} \left[ \frac{R_{mn}(k)}{\lambda^m_{\infty} + \sum_{j=1}^P \alpha_{mj} R_{mj}(k)} \right]^2,
     \end{align*}
For $ \quad n' \neq n$:
     \begin{align*}
    \frac{\partial^2 \mathcal{L}^m}{\partial \alpha_{mn} \partial \alpha_{mn'}} 
    &= - \sum_{k: t_k^m < T} \frac{R_{mn}(k) R_{mn'}(k)}{\left(\lambda^m_{\infty} + \sum_{j=1}^P \alpha_{mj} R_{mj}(k)\right)^2},
      \end{align*}
      For $\quad m' \neq m$:
      \begin{align*}
    \frac{\partial^2 \mathcal{L}^m}{\partial \beta_{mn} \partial \lambda_{\infty}^{m'}} &= 0,  \\
    \frac{\partial^2 \mathcal{L}^m}{\partial \beta_{mn} \partial \alpha_{mn}} &= 
    - \frac{1}{\beta_{mn}} \sum_{k: t_k^n < T} \Delta_{T,k}^n \exp\left(-\beta_{mn} \Delta_{T,k}^n\right) \\
    &\quad + \frac{1}{\beta_{mn}^2} \sum_{k: t_k^n < T} \left[ 1 - \exp\left(-\beta_{mn} \Delta_{T,k}^n\right) \right] 
   \\
   &\quad - \sum_{k: t_k^m < T} \frac{R'_{mn}(k)}{\lambda^m_{\infty} + \sum_{j=1}^P \alpha_{mj} R_{mj}(k)} \\
    &\quad + \sum_{k: t_k^m < T} \frac{\alpha_{mn} R_{mn}'(k) R_{mn}(k)}{\left(\lambda^m_{\infty} + \sum_{j=1}^P \alpha_{mj} R_{mj}(k)\right)^2}
    \end{align*}
  For $m' \neq m, \quad n, n' \in \{1, \dots, P\}$: 
    \begin{align*}
    \frac{\partial^2 L^m}{\partial \beta_{mn} \partial \alpha_{mn'}} &= 
    \sum_{k: t_k^m < T} \frac{\alpha_{mn} R_{mn}'(k) R_{mn'}(k)}{\left(\lambda_{\infty}^m + \sum_{j=1}^P \alpha_{mj} R_{mj}(k)\right)^2}, \\
    \frac{\partial^2 L^m}{\partial \beta_{mn} \partial \alpha_{m'n'}} &= 0,
\end{align*} 
\begin{align*}
    \frac{\partial^2 L^m}{\partial \beta_{mn}^2} &= 
    -\frac{2 \alpha_{mn}}{\beta_{mn}^3} \sum_{t_k^n < T} \left[ 1 - \exp\left(-\beta_{mn} \Delta_{T,k}^n\right) \right] \\
    &\quad + \frac{2 \alpha_{mn}}{\beta_{mn}^2} \sum_{t_k^n < T} \Delta_{T,k}^n \exp\left(-\beta_{mn} \Delta_{T,k}^n\right) \\
    &\quad + \frac{\alpha_{mn}}{\beta_{mn}} \sum_{t_k^n < T} (\Delta_{T,k}^n)^2 \exp\left(-\beta_{mn} \Delta_{T,k}^n\right) \\
    &\quad + \sum_{t_k^n < T} \Big[ \frac{\alpha_{mn} R_{mn}''(k)}{\lambda^m_{\infty} + \sum_{j=1}^P \alpha_{mj} R_{mj}(k)}\\ 
    &\quad - \left( \frac{\alpha_{mn} R_{mn}'(k)}{\lambda^m_{\infty} + \sum_{j=1}^P \alpha_{mj} R_{mj}(k)} \right)^2 \Big],
\end{align*}
where
\[
R_{mn}''(k) := \sum_{i: t_i^n < t_k^m} (\Delta_{k,i}^{mn})^2 \exp\left(-\beta_{mn} \Delta_{k,i}^{mn}\right), \quad R_{mn}''(1) = 0.
\]

\begin{align*}
\frac{\partial^2 L^m}{\partial \beta_{mn} \partial \beta_{mn'}} &= 
-\sum_{k : t_k^m < T} \frac{\alpha_{mn} R_{mn}'(k) \alpha_{mn'} R_{mn'}'(k)}{\left(\lambda_{\infty}^m + \sum_{j=1}^P \alpha_{mj} R_{mj}(k)\right)^2},\\ 
\frac{\partial^2 L^m}{\partial \beta_{mn} \partial \beta_{m'n'}} &= 0, \quad m' \neq m.
\end{align*}
\end{proposition}

\section{Goodness of fit analysis}
\label{sec:GoodnessofFit}

In this section we aim to describe how we selected the best model from the estimated point process models describing the order-flow. We will first introduce a general theory for goodness-of-fit analysis of point processes. Let the compensator (or the cumulative intensity process) of a point process $N_t$ be defined by the following integral of the conditional intensity 
\begin{equation}
    \Lambda(t) = \int_0^t \lambda(v)dv
\end{equation}

By using the compensator, the technique of random time change transforms a point process with intensity $\lambda(t)$ into a standard Poisson process. The underlying rationale is that if the proposed point process model aligns well with the observed data, the time-changed process should exhibit the defining characteristics of a standard Poisson process.

Consider a point process $N_t$ whose events occur at times $\tau_n$. According to the theorem, the transformed arrival times $\tilde{\tau}_n = \Lambda(\tau_n)$ correspond to the arrival times of a unit-rate Poisson process. As a result of the Papangelou’s theorem~\cite[Theorem 7.4.I]{daley2006introduction} , the transformed interarrival times

\begin{equation}
    \tilde{\Delta}_i = \Lambda(\tau_{i+1}) - \Lambda(\tau_i) = \int_{\tau_i}^{\tau_{i+1}} \lambda(v)\, dv
\end{equation}

are i.i.d.\ random variables following the $\text{Exp}(1)$ distribution. This finding is crucial because it allows us to verify the accuracy of a fitted point process model. By employing the random time change theorem, we can assess whether the estimated model appropriately characterizes the observed point process data.

This enables us to use statistical tests or visually check by using Quantile-Quantile plots the underlying fits of the calibrated models. The most typical statistical tests in this context are the Anderson and Darling test for the exponential distribution and the Kolmogorov-Smirnov (KS) test. The independence assumption of the compensated inter arrival times can then be tested through a Ljung-Box test for serial auto correlation. Moreover, since we are interested in evaluating different calibrated models we use some quite common techniques for model selection as the AIC score (Akaike information criterion) which is defined by $AIC = -2 \ln(L(\hat{\theta})) + 2K$, where $L$ is the likelihood function, $\hat{\theta}$ is the vector of estimated parameters of the underlying model and $K$ is the number of free parameters to be estimated. Finally, since both models have been calibrated using a Maximum Likelihood method, we also check the differences of the maximized log-likelihoods. 

\subsection{Goodness of Fit Results: Comparison Piecewise and Spline approximation}
\label{sec:GoodnessofFitResults}
The goal is to examine each of the above representations for the instantaneous intensity process of the Hawkes processes and determine in terms of a Goodness of Fit analysis.
\par 

Starting from the univariate case, which is equivalent to the model where we have a zero component in the non-diagonals of the exponential kernel matrix of the Hawkes process we examine the two different choices of the baseline function. This model essentially boils down to the model where the order-flow is uncorrelated, meaning that buy market orders only excite orders of the same sign. The first and the most naive one would be to look into hourly blocks of the LOB and for each of these blocks to fit the underlying Hawkes process in order to account for the regime switch of the baseline intensity. Of course, this naive approach is also time-consuming. Considering a total of one year of back-test data, along with the calibration of each of the 24 products and the last 8 tradable hours, this approach does not seem appealing at first sight. Hence the second approach would be to consider an exponential baseline function. However, as we described above we are now going to consider a new approach by non-parametrically estimating the baseline intensity using the idea originally introduced in the paper \cite{Chen_Hall_2013} to account for the empirical finding motivated by Figure \ref{fig:VarIntensity1}, which describes the drop of trading activity during the local trading hours and the sharp increase towards delivery for the first trading hours.  We can see from the following Figure \ref{fig:PieceWiseSpline} that both the piecewise approximation \footnote{Looking into hourly blocks of the LOB and fitting the Hawkes parameters by assuming a constant baseline intensity for each hourly block} as well as the non-parametric approximation of the baseline intensity seem to capture the increase of the market activity for both the sell and buy side of the LOB. However, it is clear from the graph that the median baseline intensities estimated using the non-parametric approach seem to better fit the drop in the market activity immediately after the start of the local trading regime, for both sides of the LOB; which is also in agreement with the empirical intensities presented in Figure~\ref{fig:VarIntensity1}. In Figure \ref{fig:PieceWiseSpline} we showcase the result only for the hourly products (16h - 18h). The total number of basis points used for the estimation of the baseline intensity is $N_{basis} = 10$. One more nice feature that we observe is that on average the shape of the baseline function for each delivery hour seems to be close. This is also showcased in Figure \ref{fig: CalibratedSpline} where we present the calibrated baseline intensity for each side of the LOB and each hourly product.

\begin{figure}[!htbp]
    \centering
    \includegraphics[width=\linewidth]{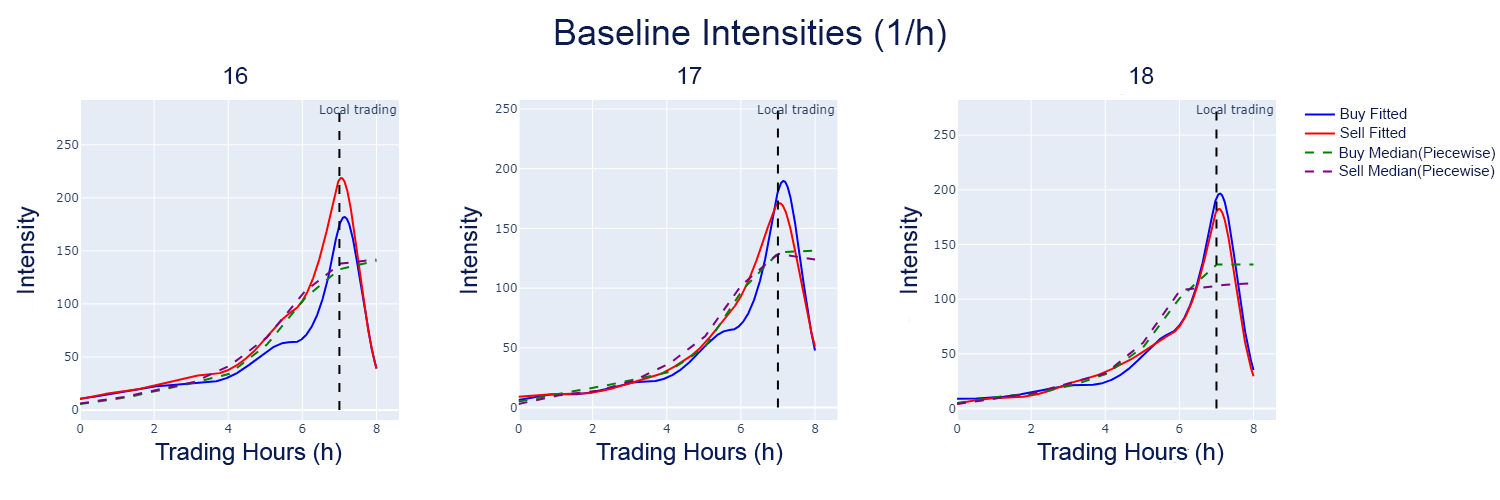}
    \caption{Basis function and piecewise (hourly blocks) approximation of the baseline intensity of the arrival rates of Buy/Sell MOs (Trading intensity)}
    \label{fig:PieceWiseSpline}
\end{figure}

\begin{figure}[!htbp]
    \centering
    \includegraphics[width=\linewidth]{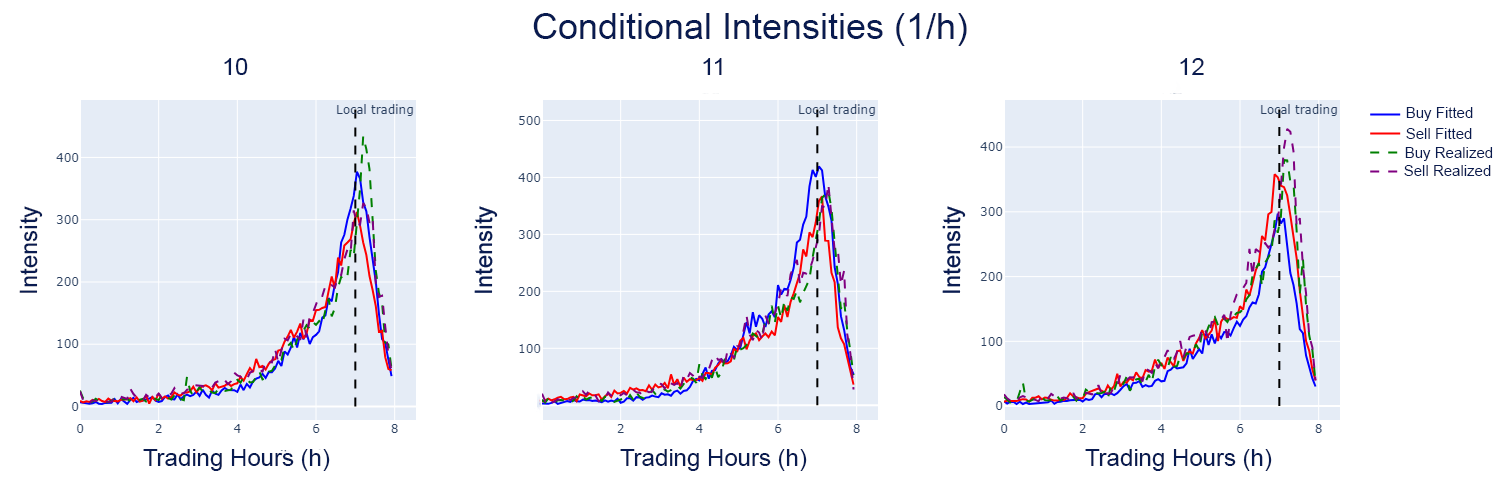}
    \caption{Empirical and fitted intensity for Buy/Sell Market order; We estimate the empirical intensity $\hat{\lambda}(t) = \mathbb{E}\left(N(t+\delta t) - N(t)\right)/\delta t$ as before. The fitted intensity is then generated by generating sequences (arrival times of Buy/Sell MOs) from the spline approximation Hawkes model and computing again the intensity based on the generated arrival times using the same methodology as described for the empirical observations.}
    \label{fig:PieceWiseSpline3}
\end{figure}

In order to evaluate the performance of each model we perform a Goodness-of Fit analysis. For the comparison of the models, we use the calibrated models for each maturity (24 in total) between the dates (2023-01-26 until 2023-12-03 a total of 311 days). Hence we have a total of 7464 models for each different approximation of the baseline intensity. Since both models have been calibrated using a Maximum Likelihood method, one way to compare the fitting performance of each model is by using a Likelihood ratio test. However, in this work, we simply use a more informal inspection of the differences in the maximum log-likelihood values of the model with a baseline piecewise constant function and the model with a spline approximation of the baseline function. Since we are calibrating each model independently we showcase the differences for each hourly contract and each daily calibration. The following Figure showcases the box plots of the difference for each hourly product and each orderbook side (Buy/Sell). For both sides and each hourly the Maximized value of the log-likelihood function for the model with piece wise constant baseline function is on average lower than the spline approximation model, which is attributed to the positive differences of the maximum likelihood values showcased in the Box plots \ref{fig:MLEdiffBuy} - \ref{fig:MLEdiffSell} \footnote{We note that the boxplots include all fitted models; (2023-01-26 until 2023-12-03 a total of 311 days) for buy and sell market orders respectively.}.

\begin{figure}[!htbp]
    \centering
    \includegraphics[width=\linewidth]{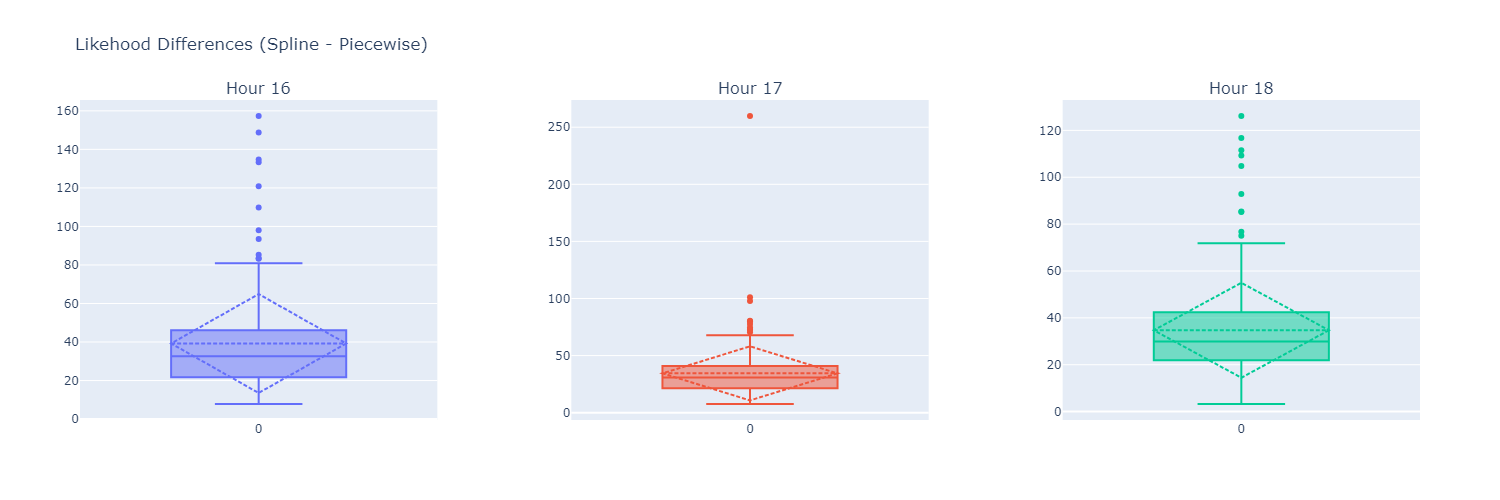}
    \caption{Difference between maximized log-likelihood of the Spline and Piecewise approximation Model (\textbf{Buy Side})}
    \label{fig:MLEdiffBuy}
\end{figure}
\begin{figure}[!htbp]
    \centering
    \includegraphics[width=\linewidth]{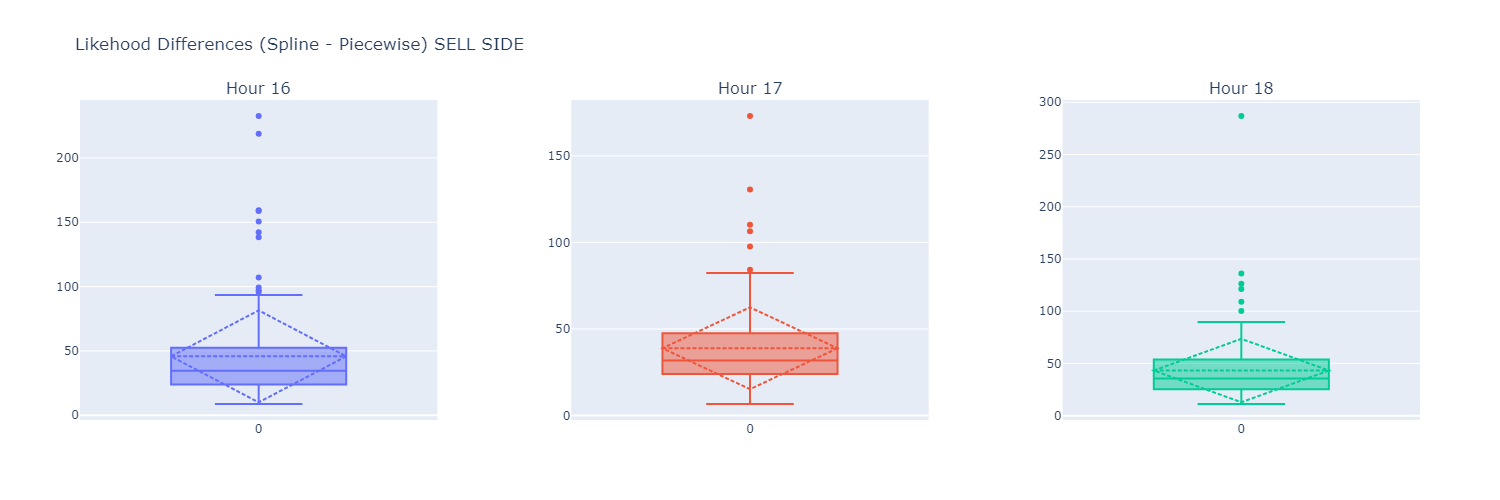}
    \caption{Difference between maximized log-likelihood of the Spline and Piecewise approximation Model (\textbf{Sell Side})}
    \label{fig:MLEdiffSell}
\end{figure}

In a similar fashion we compare the AIC differences for each hourly product.  Recall that in comparison to the log-likelihood tests the AIC scores indicate the best-performing model when the score is smaller. From the following figure, we can again see that the AIC score of the spline model is smaller (since we are testing the differences between (AIC Piecewise-  AIC Spline) and we can see a positive difference between the two) for each calibrated trading hour. This indicates that the performance of the spline fit on the baseline intensity of the Hawkes process provides a better fit to the data. This result is obtained for both sides of the LOB (Buy/Sell) (See Figures \ref{fig:AICdiffBuy}-\ref{fig:AICdiffSell}).

\begin{figure}[!htbp]
    \centering
    \includegraphics[width=\linewidth]{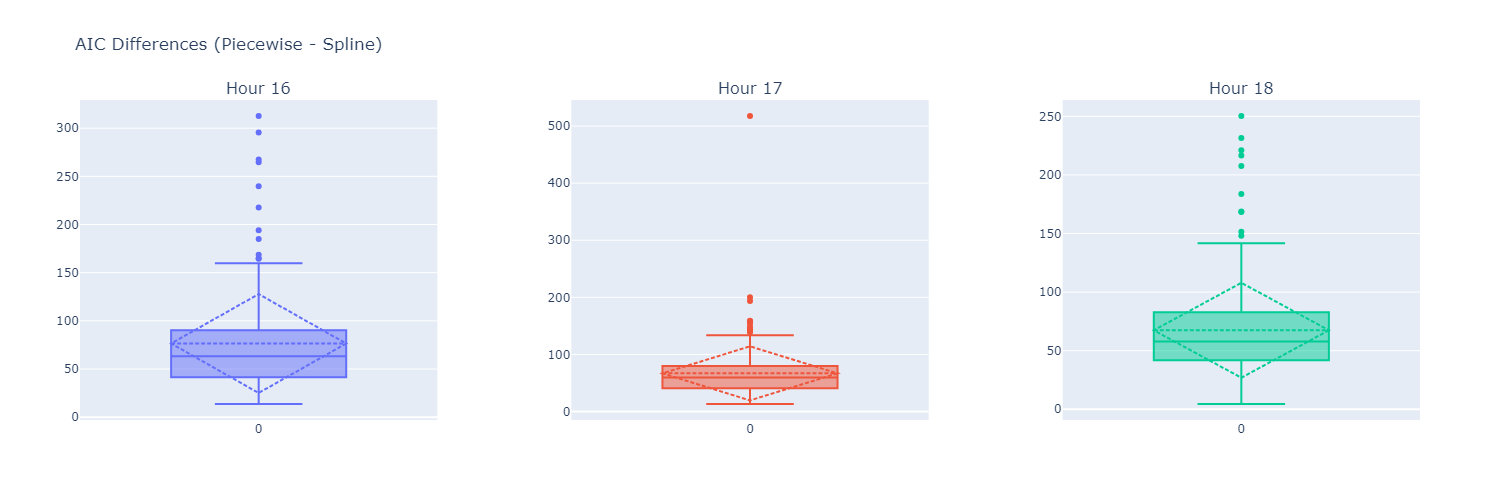}
    \caption{Differences in the AIC values - Spline Model and Piecewise approximation of the baseline intensity of the Hawkes process (\textbf{Buy Side}).}
    \label{fig:AICdiffBuy}
\end{figure}
\begin{figure}[!htbp]
    \centering
    \includegraphics[width=\linewidth]{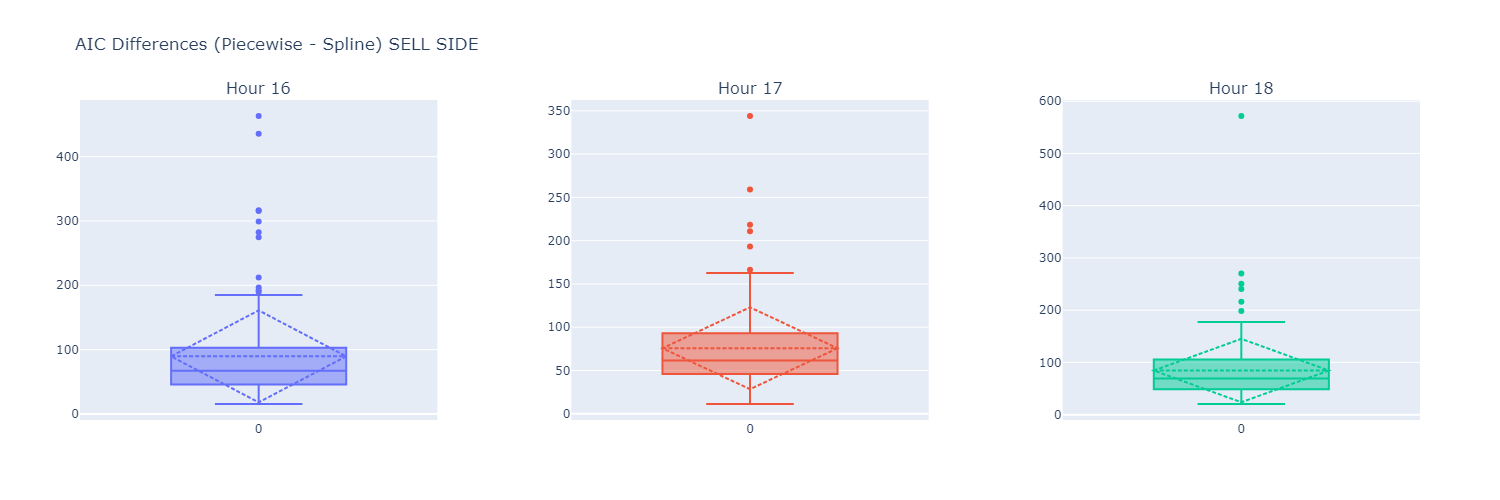}
    \caption{Differences in the AIC values - Spline Model and Piecewise approximation of the baseline intensity of the Hawkes process (\textbf{Sell Side}).}
    \label{fig:AICdiffSell}
\end{figure}

Since we have multiple calibrated models for each day and each maturity using a QQ plot between the compensated interarrival times and a unit exponential distribution would not be a very informative way to derive a qualitative conclusion on the performance of each separate model. Hence we choose to use performance metrics instead as the Anderson and Darling test for the exponential distribution and the Kolmogorov-Smirnov (KS) test. We note that in general the Anderson-Darling test is much more sensitive to the tails of the distribution, whereas KS test is more aware of the center of the distribution. One more way to test the "distance" of the time-changed inter-arrival times and the unit exponential distribution is to use a distance function defined between probability distributions on a given metric space. Suppose $\mathbb{P}$ and $\mathbb{Q}$ are two probability distributions on the real line, with corresponding cumulative distribution functions $F$ and $G$. The Wasserstein distance $W(\mathbb{P}, \mathbb{Q})$ is defined as $W(\mathbb{P}, \mathbb{Q}) = \int_{-\infty}^{\infty} |F(x) - G(x)| \, dx.$

This enables us to validate the estimated model using both the Anderson-Darling Test as well as the Wasserstein distance on an exponential distribution with parameter $\lambda = 1$ (see also \cite{Bowsher2007}). The following figures depict the Wasserstein distances for buy and sell market orders respectively and include once more all the fitted models (2023-01-26 until 2023-12-03). We only showcase the results for three maturities due to a lack of space, however, the rest of the calibrated hourly products can be found in Appendix \ref{sec:complementaryRes}. Looking at the Buy/Sell side, the median Wasserstein distance for the model with the spline approximation of the baseline intensity is significantly lower than for the piece wise approximation (see Figures \ref{fig:WassersteinBUY}-\ref{fig:WassersteinSELL}). Additionally, looking at the inter quantile range (IQR) of the distance for the case of the spline model, it is considerably lower than the piecewise approximation. The Anderson-Darling test showcases that the fit to a unit exponential is better on both sides of the LOB since the values of the test statistic are lower in both cases. Finally, we examine the independence of the time-changed interarrival times using the Ljung-Box test. The results of the p-values plot in Figures~\ref{fig:Ljung-BoxBuy}-\ref{fig:Ljung-BoxSell}, again reveal that since the mean of the p-values of all fitted models is not below $0.05$, which is the selected significance level, both tests for the spline and piecewise models agree that you can not reject the null of no auto-correlation between the series and each of its first 100 lags with > 95\% confidence level. However, testing the number of times that the null hypothesis was rejected between all fitted models (days), and examining the difference between the spline and the piece wise model, we can conclude that this number is smaller for the spline model, indicating that most of the time the time-changes inter arrival times of the spline model are indeed serially uncorrelated. 

\begin{figure}[!htbp]
    \centering
    \includegraphics[width=\linewidth]{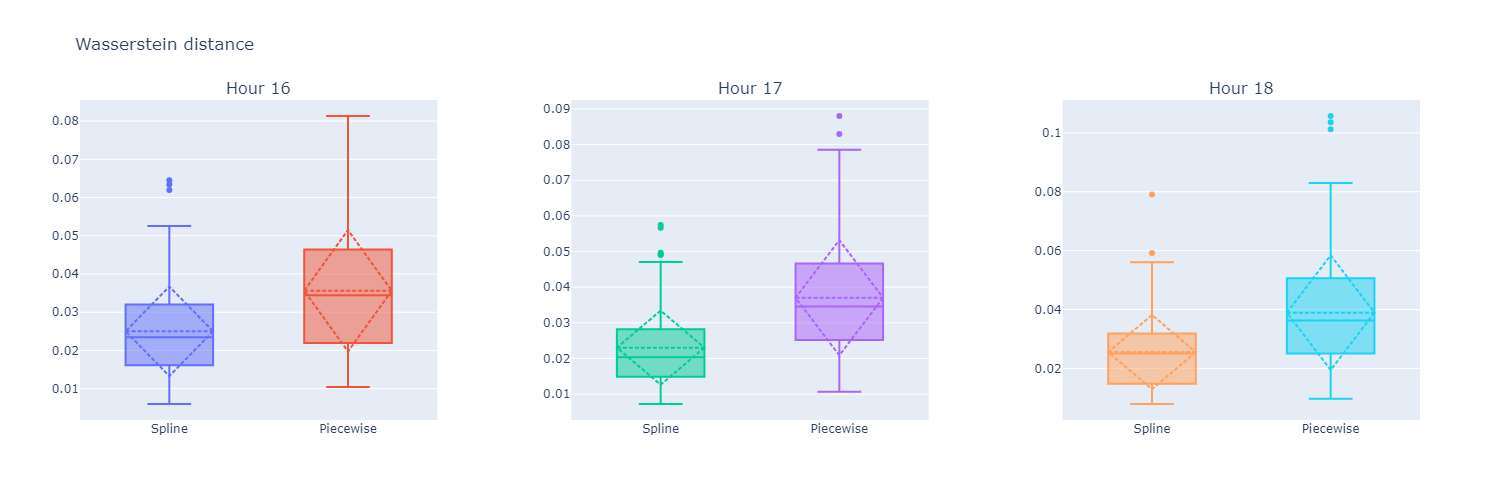}
    \caption{Wasserstein distance (\textbf{Buy Side})}
    \label{fig:WassersteinBUY}
\end{figure}

\begin{figure}[!htbp]
    \centering
    \includegraphics[width=\linewidth]{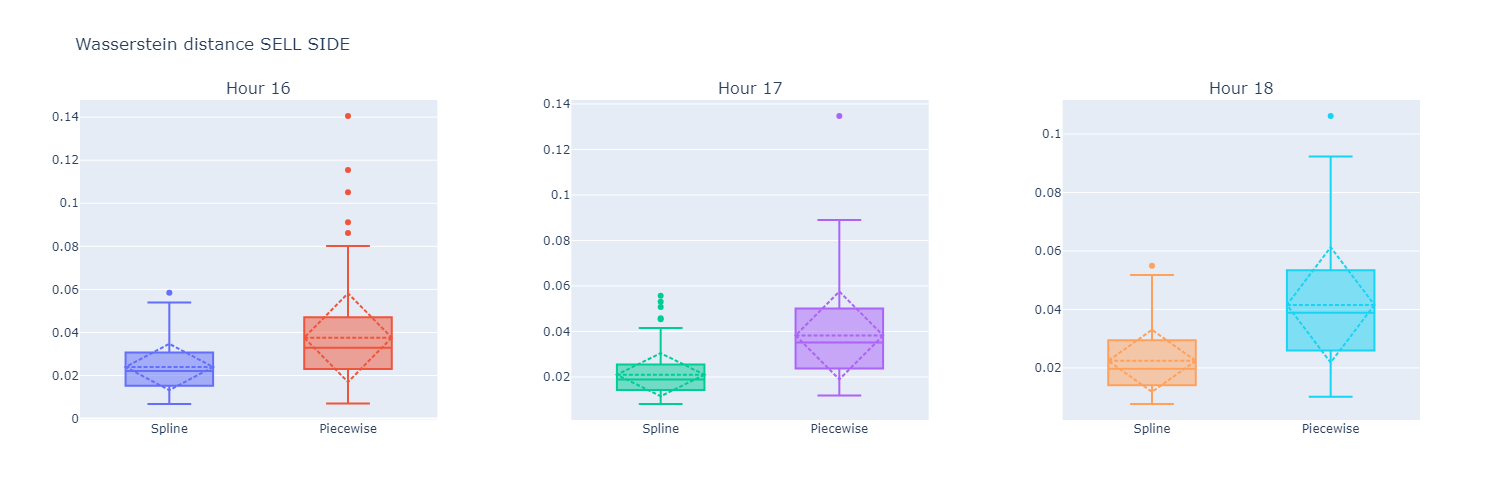}
    \caption{Wasserstein distance (\textbf{Sell Side})}
    \label{fig:WassersteinSELL}
\end{figure}

\begin{figure}[!htbp]
    \centering
    \includegraphics[width=\linewidth]{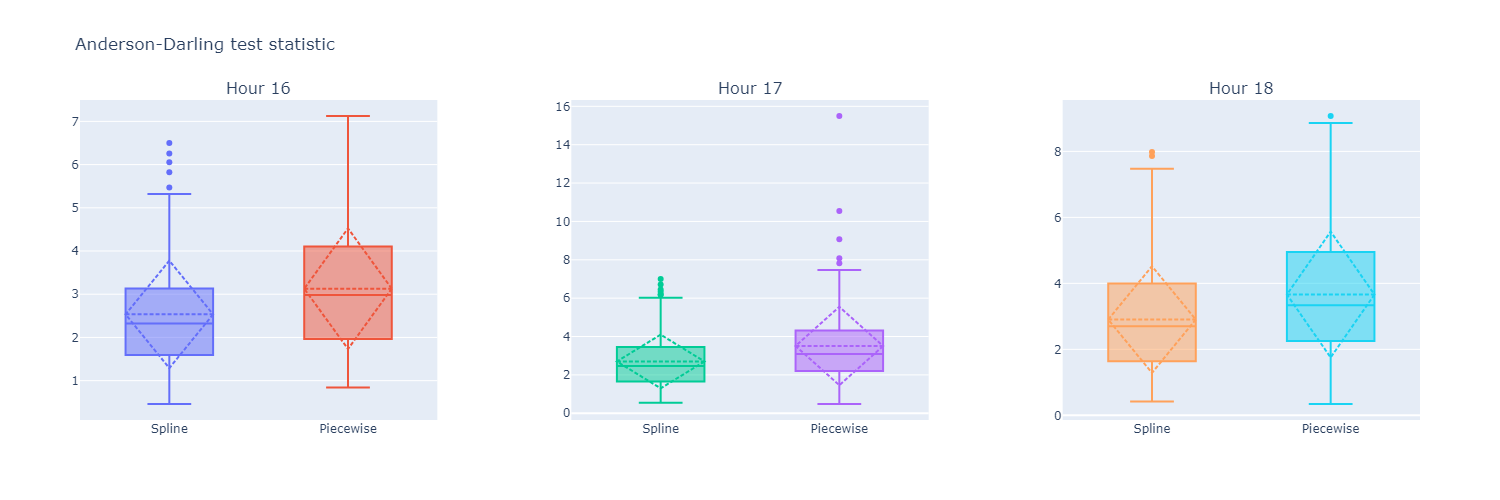}
    \caption{Results of the Anderson-Darling test statistic for the distribution of the time-changes inter arrival times compared to an Exp(1) distribution (\textbf{Buy Side}).}
\end{figure}

\begin{figure}[!htbp]
    \centering
    \includegraphics[width=\linewidth]{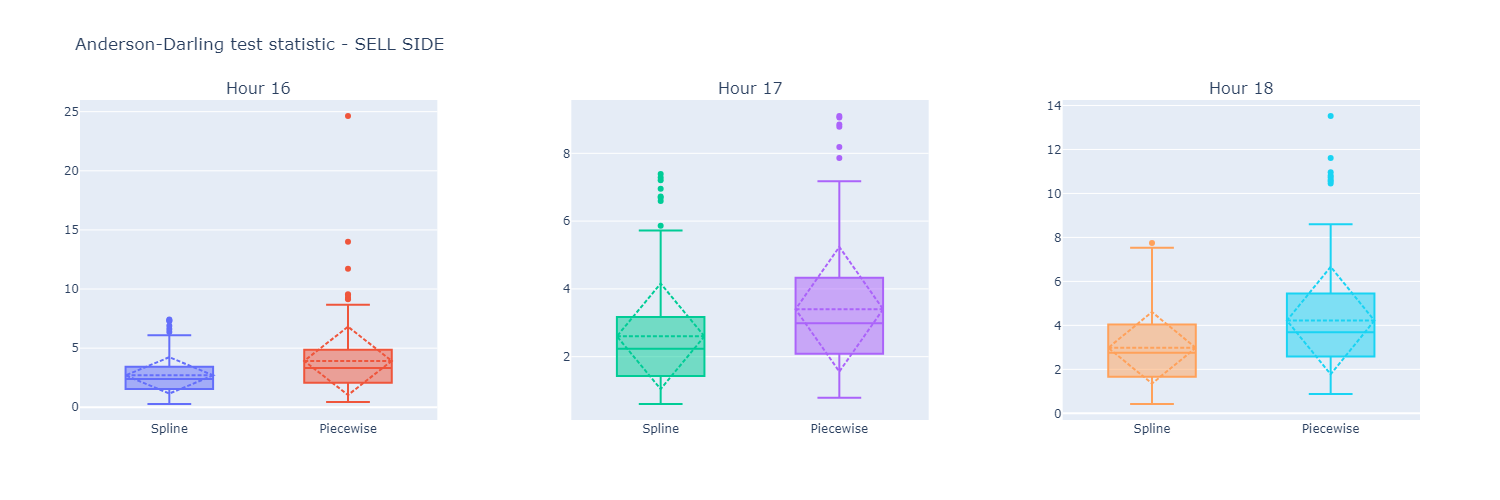}
    \caption{Results of the Anderson-Darling test statistic for the distribution of the time-changes inter arrival times compared to an Exp(1) distribution (\textbf{Sell Side}).}
\end{figure}

\begin{figure}[!htbp]
    \centering
    \includegraphics[width=\linewidth]{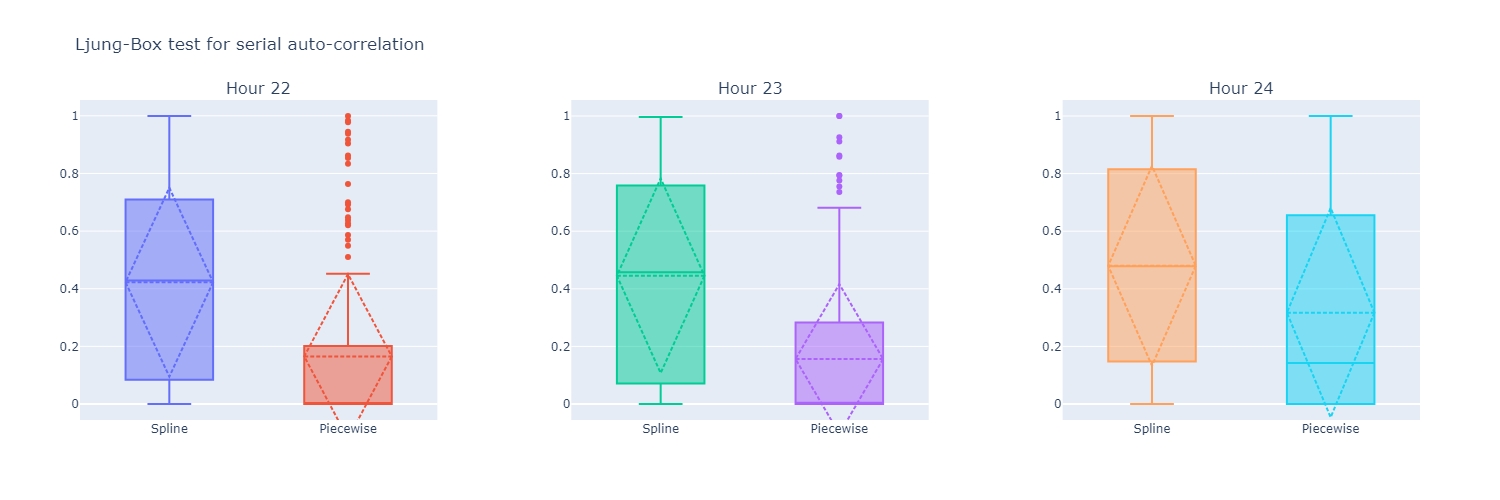}
    \caption{Results of the p-values Ljung-Box test for serial auto-correlation (\textbf{Buy Side}).}
    \label{fig:Ljung-BoxBuy}
\end{figure}

\begin{figure}[!htbp]
    \centering
    \includegraphics[width=\linewidth]{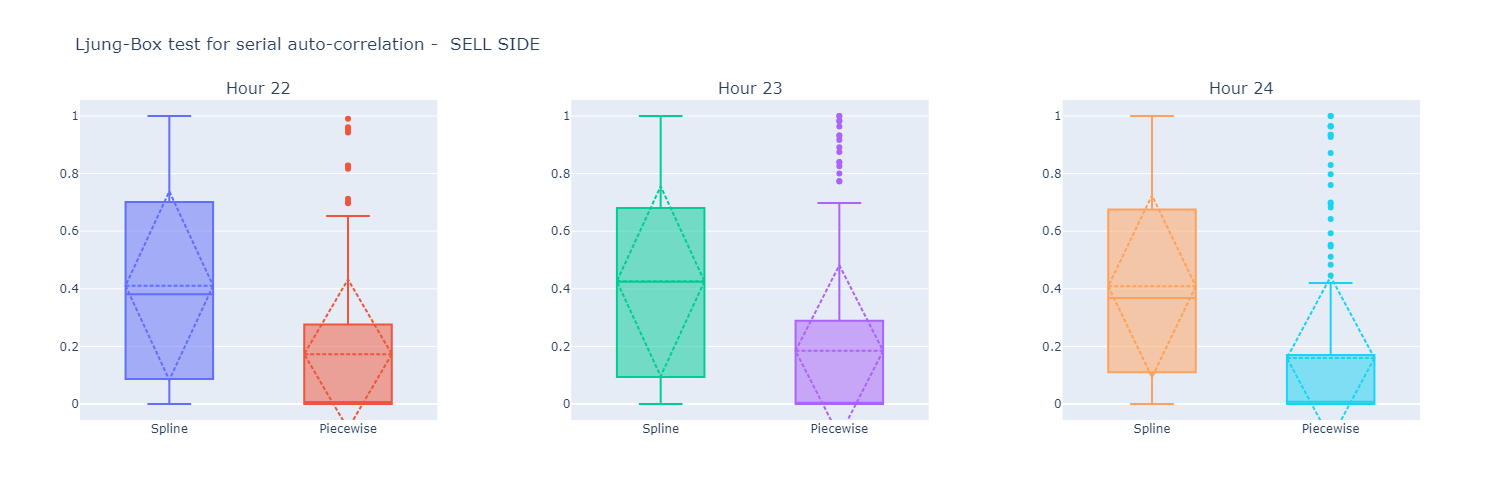}
    \caption{Results of the p-values Ljung-Box test for serial auto-correlation (\textbf{Sell Side}).}
    \label{fig:Ljung-BoxSell}
\end{figure}

\section{Complementary Results}
\label{sec:complementaryRes}
 
\begin{table}
 \centering
 \resizebox{\linewidth}{!}{
 \begin{tabular}{l|rrrr|rrrr}
 \toprule
 \textbf{Delivery Hour} & $a_{11}$ $(\frac{1}{s})$ & $a_{12}$ $(\frac{1}{s})$ & $a_{21}$ $(\frac{1}{s})$ & $a_{22}$ $(\frac{1}{s})$ & $\beta_{11}$ & $\beta_{12}$ & $\beta_{21}$ & $\beta_{22}$ \\
 \midrule
 1 & 0.14 $\pm$ 0.06 & 0.04 $\pm$ 0.03 & 0.03 $\pm$ 0.02 & 0.16 $\pm$ 0.06 & 0.63 $\pm$ 0.31 & 0.54 $\pm$ 0.29 & 0.64 $\pm$ 0.38 & 0.49 $\pm$ 0.24 \\ 
 2 & 0.12 $\pm$ 0.05 & 0.03 $\pm$ 0.02 & 0.02 $\pm$ 0.02 & 0.15 $\pm$ 0.07 & 0.55 $\pm$ 0.24 & 0.50 $\pm$ 0.22 & 0.54 $\pm$ 0.30 & 0.56 $\pm$ 0.22 \\ 
 3 & 0.12 $\pm$ 0.05 & 0.03 $\pm$ 0.03 & 0.02 $\pm$ 0.02 & 0.13 $\pm$ 0.06 & 0.55 $\pm$ 0.25 & 0.59 $\pm$ 0.33 & 0.64 $\pm$ 0.39 & 0.54 $\pm$ 0.18 \\ 
 4 & 0.13 $\pm$ 0.06 & 0.03 $\pm$ 0.02 & 0.03 $\pm$ 0.02 & 0.12 $\pm$ 0.06 & 0.62 $\pm$ 0.32 & 0.53 $\pm$ 0.18 & 0.61 $\pm$ 0.33 & 0.52 $\pm$ 0.24 \\ 
 5 & 0.13 $\pm$ 0.07 & 0.04 $\pm$ 0.03 & 0.04 $\pm$ 0.03 & 0.14 $\pm$ 0.06 & 0.59 $\pm$ 0.34 & 0.63 $\pm$ 0.34 & 0.69 $\pm$ 0.46 & 0.50 $\pm$ 0.23 \\ 
 6 & 0.15 $\pm$ 0.07 & 0.04 $\pm$ 0.03 & 0.03 $\pm$ 0.03 & 0.13 $\pm$ 0.06 & 0.64 $\pm$ 0.38 & 0.57 $\pm$ 0.31 & 0.61 $\pm$ 0.39 & 0.54 $\pm$ 0.25 \\ 
 7 & 0.15 $\pm$ 0.06 & 0.02 $\pm$ 0.02 & 0.03 $\pm$ 0.03 & 0.14 $\pm$ 0.06 & 0.69 $\pm$ 0.35 & 0.48 $\pm$ 0.24 & 0.57 $\pm$ 0.37 & 0.53 $\pm$ 0.26 \\ 
 8 & 0.13 $\pm$ 0.06 & 0.03 $\pm$ 0.03 & 0.03 $\pm$ 0.02 & 0.14 $\pm$ 0.05 & 0.59 $\pm$ 0.28 & 0.54 $\pm$ 0.26 & 0.57 $\pm$ 0.27 & 0.47 $\pm$ 0.19 \\ 
 9 & 0.14 $\pm$ 0.07 & 0.03 $\pm$ 0.02 & 0.03 $\pm$ 0.02 & 0.15 $\pm$ 0.06 & 0.70 $\pm$ 0.30 & 0.53 $\pm$ 0.22 & 0.58 $\pm$ 0.37 & 0.54 $\pm$ 0.30 \\ 
 10 & 0.13 $\pm$ 0.05 & 0.03 $\pm$ 0.02 & 0.02 $\pm$ 0.02 & 0.14 $\pm$ 0.07 & 0.61 $\pm$ 0.29 & 0.55 $\pm$ 0.31 & 0.61 $\pm$ 0.32 & 0.62 $\pm$ 0.31 \\ 
 11 & 0.13 $\pm$ 0.05 & 0.03 $\pm$ 0.02 & 0.03 $\pm$ 0.02 & 0.13 $\pm$ 0.05 & 0.59 $\pm$ 0.23 & 0.52 $\pm$ 0.18 & 0.54 $\pm$ 0.17 & 0.51 $\pm$ 0.23 \\ 
 12 & 0.14 $\pm$ 0.05 & 0.03 $\pm$ 0.02 & 0.02 $\pm$ 0.02 & 0.12 $\pm$ 0.05 & 0.57 $\pm$ 0.24 & 0.59 $\pm$ 0.30 & 0.54 $\pm$ 0.24 & 0.50 $\pm$ 0.11 \\ 
 13 & 0.12 $\pm$ 0.04 & 0.04 $\pm$ 0.03 & 0.03 $\pm$ 0.02 & 0.13 $\pm$ 0.06 & 0.53 $\pm$ 0.24 & 0.60 $\pm$ 0.36 & 0.59 $\pm$ 0.28 & 0.53 $\pm$ 0.25 \\ 
 14 & 0.13 $\pm$ 0.05 & 0.04 $\pm$ 0.02 & 0.04 $\pm$ 0.02 & 0.14 $\pm$ 0.04 & 0.52 $\pm$ 0.25 & 0.60 $\pm$ 0.39 & 0.68 $\pm$ 0.39 & 0.55 $\pm$ 0.24 \\ 
 15 & 0.12 $\pm$ 0.04 & 0.04 $\pm$ 0.02 & 0.03 $\pm$ 0.02 & 0.11 $\pm$ 0.04 & 0.49 $\pm$ 0.25 & 0.65 $\pm$ 0.42 & 0.58 $\pm$ 0.32 & 0.47 $\pm$ 0.15 \\ 
 16 & 0.12 $\pm$ 0.05 & 0.04 $\pm$ 0.02 & 0.04 $\pm$ 0.03 & 0.11 $\pm$ 0.04 & 0.55 $\pm$ 0.24 & 0.55 $\pm$ 0.26 & 0.60 $\pm$ 0.27 & 0.45 $\pm$ 0.14 \\ 
 17 & 0.12 $\pm$ 0.04 & 0.03 $\pm$ 0.02 & 0.03 $\pm$ 0.02 & 0.11 $\pm$ 0.03 & 0.48 $\pm$ 0.22 & 0.59 $\pm$ 0.35 & 0.54 $\pm$ 0.30 & 0.46 $\pm$ 0.18 \\ 
 18 & 0.11 $\pm$ 0.05 & 0.03 $\pm$ 0.02 & 0.03 $\pm$ 0.02 & 0.12 $\pm$ 0.04 & 0.52 $\pm$ 0.24 & 0.59 $\pm$ 0.36 & 0.61 $\pm$ 0.36 & 0.50 $\pm$ 0.13 \\ 
 19 & 0.13 $\pm$ 0.05 & 0.03 $\pm$ 0.02 & 0.03 $\pm$ 0.03 & 0.12 $\pm$ 0.06 & 0.52 $\pm$ 0.23 & 0.58 $\pm$ 0.39 & 0.61 $\pm$ 0.37 & 0.48 $\pm$ 0.21 \\ 
 20 & 0.12 $\pm$ 0.04 & 0.04 $\pm$ 0.03 & 0.04 $\pm$ 0.02 & 0.14 $\pm$ 0.06 & 0.56 $\pm$ 0.28 & 0.57 $\pm$ 0.32 & 0.57 $\pm$ 0.28 & 0.47 $\pm$ 0.15 \\ 
 21 & 0.13 $\pm$ 0.05 & 0.03 $\pm$ 0.02 & 0.04 $\pm$ 0.03 & 0.14 $\pm$ 0.05 & 0.60 $\pm$ 0.29 & 0.58 $\pm$ 0.34 & 0.58 $\pm$ 0.27 & 0.50 $\pm$ 0.20 \\ 
 22 & 0.12 $\pm$ 0.05 & 0.03 $\pm$ 0.02 & 0.03 $\pm$ 0.03 & 0.13 $\pm$ 0.05 & 0.54 $\pm$ 0.27 & 0.55 $\pm$ 0.32 & 0.58 $\pm$ 0.32 & 0.46 $\pm$ 0.18 \\ 
 23 & 0.13 $\pm$ 0.06 & 0.03 $\pm$ 0.02 & 0.03 $\pm$ 0.03 & 0.12 $\pm$ 0.05 & 0.59 $\pm$ 0.29 & 0.58 $\pm$ 0.31 & 0.55 $\pm$ 0.28 & 0.47 $\pm$ 0.14 \\ 
 24 & 0.12 $\pm$ 0.05 & 0.03 $\pm$ 0.02 & 0.03 $\pm$ 0.02 & 0.13 $\pm$ 0.04 & 0.52 $\pm$ 0.24 & 0.49 $\pm$ 0.27 & 0.59 $\pm$ 0.30 & 0.48 $\pm$ 0.19 \\ 
 \bottomrule
 \end{tabular}}
 \caption{Medians ($\pm$ std) of the self/cross-excitation and exponential decay parameters for all delivery hours.}
 \label{tab:combined-table}
 \end{table}

\begin{figure}[!htbp]
    \centering
    \includegraphics[width=\linewidth]{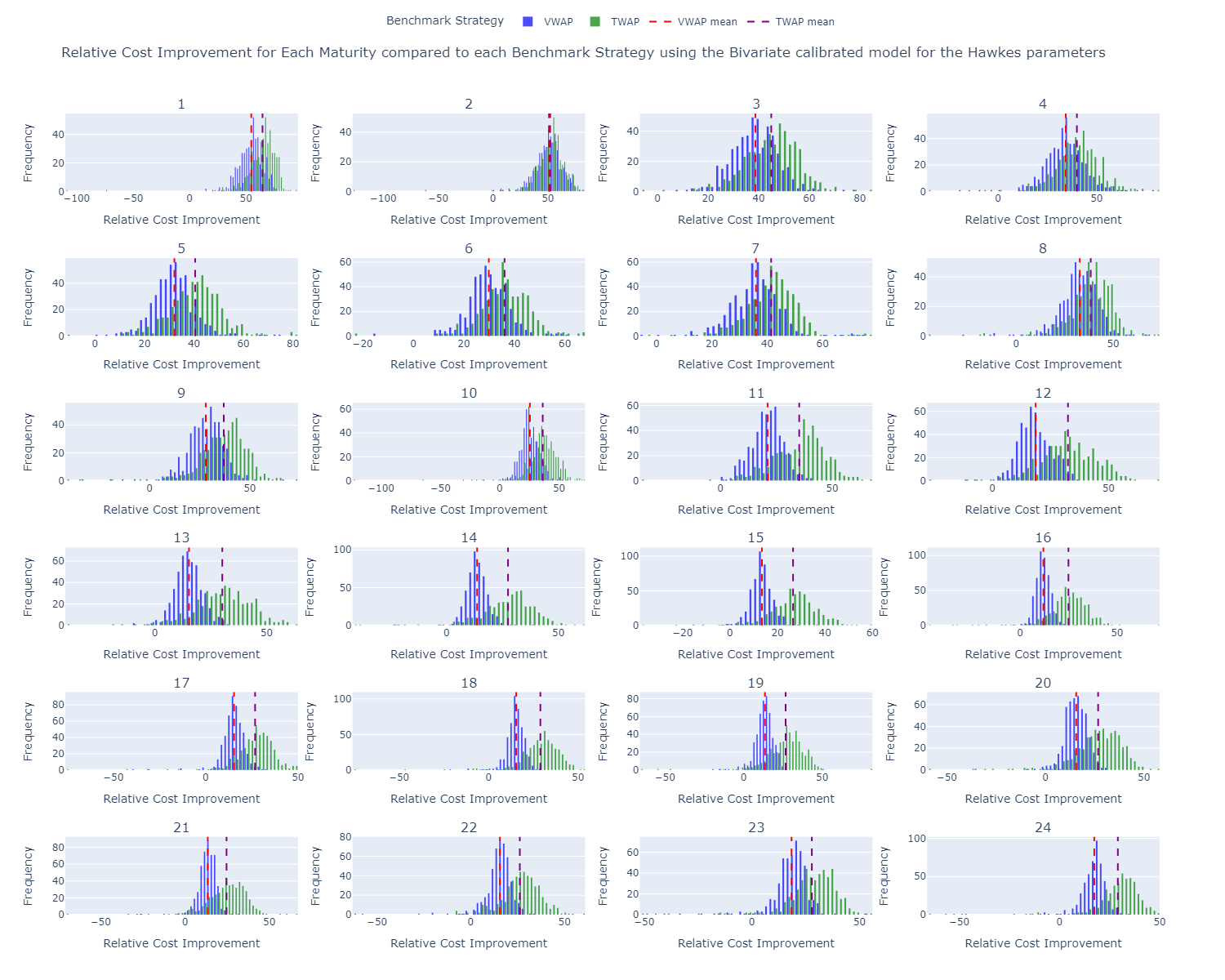}
    \caption{Relative Cost Improvement in terms of cumulative costs (\%) for each hourly product $r_{BenchMark}$ (bivariate).}
    \label{fig:RelativeCostImprov2}
\end{figure}

\begin{figure*}[!htbp]
    \centering
    \begin{subfigure}[b]{0.45\textwidth}
    \centering
    \includegraphics[width=\textwidth]{Figures/newplot_-_2024-07-31T144838.060__1__copy_4.png}
    \end{subfigure}
\hfill
    \begin{subfigure}[b]{0.45\textwidth}
    \centering
    \includegraphics[width=\textwidth]{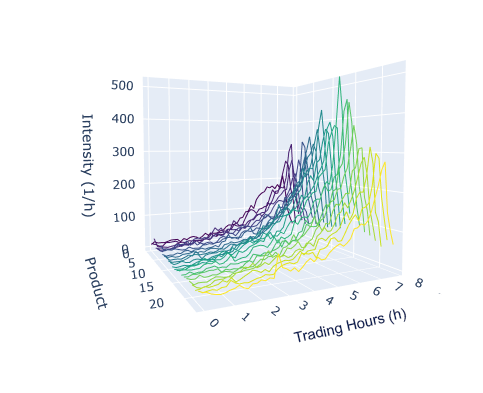}
    \end{subfigure}
    \caption{Empirical intensities of Buy/Sell Market orders (1/h) (Mean across all trading sessions) calculated for each product using $\hat{\lambda}(t) = \mathbb{E}\left(N(t+\delta t) - N(t)\right)/\delta t$  from the number of realizations of Buy/Sell MOs through counting the average number of events during $[t,t+\delta t]$, where $N(t)$ is the counting process for $\lambda(t)$.}
    \label{fig: EmpiricalBaseline}
\end{figure*}\noindent

\end{document}